\documentclass[a4paper]{llncs}
\usepackage[x11names, rgb]{xcolor}
\usepackage[utf8]{inputenc}
\usepackage{tikz}
\usetikzlibrary{snakes,arrows,shapes}
\tikzset{every node/.style={inner sep=.5pt}}
\tikzset{every picture/.style={scale=.6}}
\usepackage{subfig}
\usepackage{cite}
\usepackage{amssymb}
\usepackage{empheq} 
\usepackage{url}
\usepackage{color}
\usepackage[ruled,vlined]{algorithm2e}

\usepackage{graphicx}
\usepackage{wrapfig}
\usepackage{hyperref}

\DeclareMathOperator{\inedges}{in}
\DeclareMathOperator{\outedges}{out}
\DeclareMathOperator{\exe}{Exe}
\DeclareMathOperator{\exep}{\exe(\protspec)}
\DeclareMathOperator{\exepclosed}{\exe_{C}(\protspec)}

\newcommand{\abs}[1]{\left| #1\right|}
\newcommand{\edges}{\ensuremath{E}} 
\newcommand{\vertices}{\ensuremath{V}} 
\newcommand{\paths}{\mathrm{path}} 
\newcommand{\seq}{\mathrm{seq}} 
\newcommand{\maxset}{\mathrm{maxset}} 
\newcommand{\minset}{\mathrm{minset}} 
\newcommand{\set}[1]{\left\{ #1\right\}}
\newcommand{\knowledge}{\ensuremath{K}}
\newcommand{\preknowledge}{\ensuremath{K_{\prec}}}
\newcommand{\hon}{\mathrm{hon}} 
\newcommand{\Sigset}{\mathit{SigSet}} 
\newcommand{\Initset}{\mathit{InitSet}} 
\newcommand{\Endset}{\mathit{EndSet}} 
\newcommand{\evidence}{\mathit{Evidence}}
\newcommand{\dishonest}{\mathit{Dishonest}}
\newcommand{\decision}{\mathit{decision}}
\newcommand{\abort}{\text{``abort''}}

\newcommand{\promises}{\wp} 
\newcommand{\promise}[5]{\promises_{#1}(#2,#3,#4,#5)} 
\newcommand{\signature}{\mathcal{S}} 

\newcommand{\res}{\mathit{Res}}
\newcommand{\causal}{\varepsilon}
\newcommand{\state}{\mathit{s}}
\newcommand{\states}{\mathcal{S}}
\newcommand{\protspec}{\mathcal{P}}
\newcommand{\rec}{\mathit{recv}}
\newcommand{\snd}{\mathit{send}}
\newcommand{\exit}{\mathit{exit}}
\newcommand{\exitrole}{\ttp}
\newcommand{\ttpfun}{\ensuremath{\delta}}

\newcommand{\signers}{\ensuremath{A}} 
\newcommand{\signer}{P} 
\newcommand{\othersigner}{Q} 
 
\newcommand{\contract}{\mathfrak{C}}

\newcommand{\ttp}{\ensuremath{\mathrm{T}}} 

\newcommand{\naturals}{{\mathbb{N}}}

\newcommand{\roles}{R}
\newcommand{\mesg}{M}
\newcommand{\statessub}[1]{\states_{#1}}
\newcommand{\statesp}{\statessub{\protspec}}
\newcommand{\nlabel}{r} 
\newcommand{\elabel}{\mu}

\newcommand{\transition}[3]{#1\stackrel{#2}{\leadsto}#3}

\begin{document}
\title{Generalizing Multi-party Contract Signing\thanks{This is the
 extended version of~\cite{MR15}.}}
\author{Sjouke Mauw\inst{1} \and Sa\v{s}a Radomirovi\'{c}\inst{2}}
\institute{CSC/SnT, University of Luxembourg\\
\email{sjouke.mauw@uni.lu}
\and 
Institute of Information Security,
Dept.~of Computer Science,
ETH Zurich\\
\email{sasa.radomirovic@inf.ethz.ch}
}

\maketitle

\begin{abstract}
Multi-party contract signing (MPCS) protocols allow a group of signers
to exchange signatures on a predefined contract. 
Previous approaches considered either completely linear protocols or
fully parallel broadcasting protocols.
We introduce the new class of DAG
MPCS protocols which combines parallel and linear execution 
and allows for parallelism even within a signer role.
This generalization is useful in practical applications where the set
of signers has a hierarchical structure, such as 
chaining of service level agreements and subcontracting. 

Our novel DAG MPCS protocols are represented by directed acyclic graphs
and equipped with a labeled transition system semantics.
We define the notion of \emph{abort-chaining sequences} and prove that
a DAG MPCS protocol satisfies fairness if and only if it does not
have an abort-chaining sequence. We exhibit several examples of 
optimistic fair DAG MPCS protocols. The fairness of these
protocols follows from our theory and has additionally been verified
with our automated tool.

We define two complexity measures for DAG MPCS protocols, related
to execution time and total number of messages exchanged. We prove 
lower bounds for fair DAG MPCS protocols in terms of these
measures.

\end{abstract}

\section{Introduction}\label{sec:intro}

A multi-party contract signing (MPCS) protocol is a communication
protocol that allows a number of parties to sign a digital contract. 
The need for MPCS protocols arises, for instance, in the context of
service level agreements (SLAs) and in supply chain contracting. 
In these domains (electronic) contract negotiations and signing are
still mainly bilateral. 
Instead of
negotiating and signing one multi-party contract, in practice,
multiple bilateral negotiations are conducted in
parallel~\cite{ywk11}. Because 
negotiations can fail, parties may end up with just a subset
of the pursued bilateral contracts. If a party is missing contracts with providers or
subcontractors, it faces an 
\emph{overcommitment} problem. If contracts with customers are missing,
it has an \emph{overpurchasing} problem~\cite{kk10}. Both
problems can be prevented by using fair multi-party contract signing
protocols.

Existing optimistic MPCS protocols come in two flavors. \emph{Linear} MPCS
protocols require that at any point in time at most one signer has
enough information to proceed in his role by sending messages to other
signers. \emph{Broadcast}
MPCS protocols specify a number of
communication rounds in each of which all signers send or broadcast
messages to each other. 
However, neither of the two kinds of protocols is suitable for SLAs or
supply chain contracting.
The reason is that in both domains, the set of
contractors typically has a hierarchical structure, consisting of main
contractors and levels of subcontractors. It
is undesirable (and perhaps even infeasible) for the main contracting
partners and their subcontractors to directly communicate with
another partner's subcontractors.
This restriction immediately excludes broadcast protocols as
potential solutions and forces linear protocols to be impractically
large. 

In this paper we introduce MPCS protocol specifications that support
arbitrary combinations of linear and parallel actions, even within a
protocol role. 
The message flow of such protocols can be specified as a directed
acyclic graph (DAG) and we therefore refer to them as \emph{DAG} MPCS
protocols. 

A central requirement for MPCS protocols is \emph{fairness}. 
This 
means that either all honest
signers get all signatures on the negotiated contract or nobody gets the honest
signers' signatures. 
It is well known that in asynchronous communication networks, 
a deterministic MPCS protocol requires
a trusted third party (TTP) to achieve fairness~\cite{technion}.
In order to prevent the TTP from becoming a bottleneck, protocols have
been designed in which the TTP is only involved to resolve conflicts.
A conflict may occur if a message gets lost, if
an external adversary interferes with the protocol, or if signers do not 
behave according to the protocol specification.
If no conflicts occur, the TTP does not even have to be
aware of the execution of the protocol. Such protocols are called
\emph{optimistic}~\cite{asokanphd}.
We focus on optimistic protocols in this paper.

DAG MPCS protocols not only allow for better solutions to the
subcontracting problem, but also 
have 
further advantages over linear and
broadcast MPCS protocols and 
we design three novel MPCS protocols that demonstrate this. 
One such advantage concerns communication complexity. Linear protocols 
can reach the minimal number of messages necessary to be exchanged in
fair MPCS protocols at the cost of a high number of protocol
``rounds''. We call this the \emph{par\-al\-lel complexity}, which is a
generalization of the round complexity measure for broadcast protocols, and
define it in Section~\ref{sec:complexity}.
Conversely, broadcast protocols can attain the minimal number
of protocol 
rounds necessary for fair MPCS, but at the cost of a high message complexity. 
We demonstrate that DAG MPCS protocols can simultaneously attain best
possible order of magnitude for both complexity measures. 

As discussed in our related work section, the design of fair MPCS
protocols has proven to be non-trivial and error-prone. 
We therefore not only prove our three novel DAG MPCS protocols to be
fair, but we also derive necessary and sufficient conditions for
fairness of any optimistic DAG MPCS protocol. These conditions can be
implemented and verified automatically, but they are still 
non-trivial. Therefore, for a slightly restricted class of DAG protocols,
we additionally derive a fairness criterion that is easy to verify.

\textbf{Contributions.} 
Our main contributions are
(i) the definition of a syntax and interleaving semantics of DAG
MPCS protocols (Section~\ref{sec:spec_exec_model});
(ii) the definition of the message complexity and parallel complexity
of such protocols (Section~\ref{sec:complexity});
(iii) a method to derive a full MPCS specification from a
\emph{skeletal graph}, including the TTP logic (Section~\ref{sec:MPCS});
(iv) 
necessary and sufficient 
conditions for fairness of DAG MPCS protocols 
(Section~\ref{sec:fairness});
(v) minimal complexity bounds for DAG MPCS protocols (Section~\ref{sec:minimal_complexity});
(vi) novel fair MPCS protocols (Section~\ref{sec:constructions});
(vii) a software tool that verifies whether a given MPCS protocol is
fair{} (described in Appendix~\ref{s:tool}).
\footnote{The tool and models of our protocols 
are available at the following website:
\url{http://people.inf.ethz.ch/rsasa/mpcs}}

\section{Related Work}

We build on the body of work that has been published
in the field of fair optimistic MPCS
protocols in asynchronous networks. The first such protocols were
proposed by Baum-Waidner and Waidner~\cite{BW98}, viz.~a
round-based broadcast protocol and a related round-based linear
protocol. They 
showed subsequently~\cite{bw00} that these
protocols are round-optimal. This is a complexity measure
that is related to, but less general than, parallel complexity
defined in the present paper. 

Garay et al.~\cite{gjm99} introduced the notion of
\emph{abuse-free} contract signing. They developed the technique of
\emph{private contract signature} and used it to create
abuse-free two-party and three-party contract signing protocols. Garay
and Mac\-Kenzie~\cite{garay99} proposed MPCS protocols
which were later shown to be unfair using the model checker Mocha 
and improved by Chadha et al.~\cite{cks04}. 
Mukhamedov and Ryan~\cite{mr08} developed the notion of \emph{abort
chaining attacks} and used such attacks to show that Chadha et al.'s
improved version does not satisfy fairness in cases where there are
more than five signers. They introduced a new optimistic MPCS protocol and proved fairness for their protocol
by hand and used the NuSMV model checker to verify the case of five signers.
Zhang et al.~\cite{ZZPM12} have used 
Mocha 
to analyze the protocols of Mukhamedov and Ryan and of Mauw et
al.~\cite{MRT09}. 

Mauw et al.~\cite{MRT09} used the notion of abort chaining to
establish a lower bound on the message complexity of linear fair MPCS
protocols. This complexity measure is generalized in the present paper
to DAG MPCS protocols. 
Kordy and Radomirovi\'c~\cite{KR12} have shown 
an explicit construction for fair linear MPCS protocols. The 
construction covers in particular the protocols proposed by Mukhamedov and
Ryan~\cite{mr08} and the linear protocol of Baum-Waidner and
Waidner~\cite{BW98}, but not the broadcast protocols. The DAG MPCS
protocol model and fairness results 
developed in the present paper encompass both types of
protocols.
They allow for
arbitrary combinations of linear and parallel behaviour (i.e.\ partial
parallelism), and in addition allow for parallelism within signer
roles.
MPCS protocols combining linear and 
parallel behaviour
have not been studied yet.

Apart from new theoretical insights to be gained from designing and
studying DAG MPCS protocols, we anticipate interesting
application domains in which multiple parties establish a number of
related contracts, such as SLAs.
Emerging business models like Software as a Service 
require a negotiation to balance a customer's requirements against a service
provider's capabilities. 
The result of such a negotiation is
often complicated by the dependencies between several
contracts~\cite{lyyk12} and multi-party protocols may serve to
mitigate this problem. 
Karaenke and Kirn~\cite{kk10} propose a multi-tier negotiation
protocol to mitigate the problems of overcommitment and 
overpurchasing.
They formally verify that the protocol solves the two
observed problems, but do not consider the fairness problem.
SLAs and negotiation protocols have also been studied in the
multi-agent community. An example is the work of
Kraus~\cite{kraus2001} who defines a multi-party negotiation
protocol in which agreement is reached if all agents accept an offer.
If the offer is rejected by at least one agent, a new offer will be
negotiated.

Another interesting application area concerns \emph{supply chain
contracting}~\cite{KrWi12}. 
A supply chain consists of a series of firms involved in the
production of a product or service with potentially complex contractual relationships.
Most literature in this area focuses on economic aspects,
like pricing strategies.
An exception is the recent work of Pavlov and Katok~\cite{KaPa13} in
which fairness is studied from a game-theoretic point of view.
The study of multi-party signing protocols and multi-contract
protocols has only recently been identified as an interesting research
topic in this application area~\cite{SeZeLi12}.

\section{Preliminaries}

\subsection{Multi-party contract signing}

The purpose of a multi-party contract signing
protocol is to allow a number of parties to sign a digital contract in
a fair way. In this section we recall the basic notions
pertaining to MPCS protocols. 
We use $\signers$ to denote the set of signers involved in a protocol,
 $\contract$ to denote the contract, and 
$\ttp$ to denote the TTP.

A signer is considered \emph{honest} (cf.\
Definition~\ref{def:honestSigner}) if it faithfully executes the
protocol specification. 
An MPCS protocol is said to be \emph{optimistic} if its execution in
absence of adversarial behaviour and failures and with all honest
signers results in signed contracts for all participants 
without any involvement of $\ttp$.
Optimistic MPCS protocols consist of two subprotocols: the
\emph{main}
protocol that specifies the exchange of \emph{promises} and \emph{signatures} by the
signers, and the \emph{resolve} protocol that describes the
interaction between the agents and $\ttp$ in case of a failure in the
main protocol.
A promise made by a signer indicates the intent
to sign $\contract$. 
A promise $\promise{\signer}{m}{x}{\othersigner}{\ttp}$ can only be
 generated by signer $\signer\in\signers$. The content $(m,x)$ can be
 extracted from the promise and the promise can be verified by signer
 $\othersigner\in\signers$ and by $\ttp$. 
A signature $\signature_\signer(m)$ can only be generated
 by $\signer$ and by $\ttp$, if $\ttp$ has a promise
 $\promise{\signer}{m}{x}{\othersigner}{\ttp}$.
The content $m$ can be extracted and the signature can be verified by anybody.
Cryptographic schemes that allow for the above properties are digital
signature schemes and private contract signatures~\cite{gjm99}.

MPCS protocols must satisfy at least two security requirements, namely
\emph{fairness} and \emph{timeliness}.
An optimistic MPCS protocol for contract $\contract$ 
is said to be \emph{fair} 
for an honest signer $\signer$ 
if whenever some signer $\othersigner\not=\signer$ obtains a 
signature on $\contract$ 
from $\signer$, then $\signer$ can obtain a signature on $\contract$
from all signers participating in the protocol.
An optimistic MPCS protocol is said to satisfy 
\emph{timeliness}, if each signer has a recourse to stop endless 
waiting for expected messages.
The fairness requirement will be the guiding principle for our
investigations and timeliness will be implied by the communication
model together with the behaviour of the TTP. 
A formal definition of fairness is given in 
Section~\ref{sec:fairness}.

A further desirable property for MPCS protocols is abuse-freeness
which was introduced in~\cite{gjm99}. 
An optimistic MPCS protocol is said to be \emph{abuse-free},
if it is impossible
for any set of signers at any point in the protocol to be able 
to prove to an outside party that they have the power to 
terminate or successfully complete the contract signing. 
Abuse-freeness is outside the scope of this paper.

\subsection{Graphs}

Let $G = (\vertices,\edges)$ with
$\edges\subseteq\vertices\times\vertices$ be a directed acyclic graph.
Let $v,w\in \vertices$ be vertices. We say that $v$ \emph{causally
 precedes} $w$, denoted $v\prec w$, if there
is a directed path from $v$ to $w$ in the graph. We write $v\preceq w$
for $v\prec w \vee v=w$. 
We extend \emph{causal precedence} to the set $\vertices\cup\edges$ as follows.
Given two edges $(v,w), (v',w')\in\edges$, we say that 
$(v,w)$ \emph{causally precedes} $(v',w')$ and write $(v,w)\prec (v',w')$,
if $w\preceq v'$. Similarly, we write $v\prec (v',w')$ if $v\preceq v'$
and $(v,w)\prec v'$ if $w\preceq v'$. 
Let $x,y\in\vertices\cup\edges$.
If $x$ causally precedes $y$ we 
also say that $y$ \emph{causally follows} $x$. 
We say that a set $S\subseteq\vertices\cup\edges$ is
\emph{causally closed} if it contains all causally preceding vertices
and edges of its elements, i.e., 
$\forall x\in S, y\in\vertices\cup\edges: y\prec x\implies
y\in S$.

By $\inedges(v)\subseteq \edges$ we denote 
the set of edges incoming to $v$ and by $\outedges(v)\subseteq \edges$
the set of edges outgoing from $v$. Formally, we have 
$\inedges(v)=\set{(w,v)\in\edges\mid w\in\vertices}$
and 
$\outedges(v)=\set{(v,w)\in\edges\mid w\in\vertices}$.

\subsection{Assumptions}\label{sec:assumptions}

The communication between signers is asynchronous and messages can get lost or be
delayed arbitrary long. 
The communication channels between signers and the TTP $\ttp$ are 
assumed to be \emph{resilient}. This means that the messages sent
over these channels are guaranteed to be delivered eventually. 
In order to simplify our reasoning, 
we assume that the channels between protocol participants
are confidential and authentic. 
We consider the problem of delivering confidential and
authentic messages in a Dolev-Yao intruder model 
to be orthogonal to the present problem setting. 

We assume that $\contract$ contains
the contract text along with fresh values
(contributed by every signer) which
prevent different protocol executions from generating 
interchangeable protocol messages.
Furthermore we assume that $\contract$ contains all information
that $\ttp$ needs in order
to 
reach a decision regarding the contract 
in case it is contacted by a signer. This information contains the protocol
specification, an identifier for $\ttp$, identifiers for the signers
involved in the protocol, and the 
assignment of signers to protocol roles in the protocol specification.

We assume the existence of a designated resolution process per signer
which coordinates the various resolution requests of the signer's
parallel threads. It will ensure that $\ttp$ is contacted at most
once by the signer.
After having received the first request from one of the signer's
threads, this resolution process will contact $\ttp$ on behalf of the
signer and store $\ttp$'s reply. This reply will be forwarded
to all of the signer's threads whenever they request resolution.

\section{DAG Protocols} \label{sec:parallel_protocols}

Our DAG protocol model is a multi-party protocol model in an
asynchronous network with a TTP and an
adversary that controls a subset of parties. 
\subsection{Specification and Execution Model}\label{sec:spec_exec_model}

A \emph{DAG protocol specification} (or simply, a \emph{protocol
specification}) is a directed acyclic graph in which the vertices
represent the state of a signer and the edges represent either a causal
dependency between two states (an $\causal$-edge) or the sending
of a message.
A vertex' outgoing edges can be executed in parallel.
Edges labeled with $\exit$ denote that a signer contacts
$\exitrole$. 

\begin{definition}
\label{def:parallel_prot_spec}
Let $\roles$ be a set of roles 
such that $\exitrole\not\in\roles$ and 
$\mesg$ a set of messages. Let
$\causal$ and $\exit$ be two symbols, such that 
$\causal,\exit\notin\mesg$. 
By $\mesg_{\causal}^{\exit}$ and $\roles_{\exitrole}$ we denote the sets 
$\mesg_{\causal}^{\exit}=\mesg\cup\{\causal, \exit\}$
 and
$\roles_{\exitrole}=\roles\cup\set{\exitrole}$, respectively.
A \emph{DAG protocol specification} 
is a labeled directed acyclic graph
$\protspec=(\vertices,\edges,\nlabel,\elabel,\ttpfun)$, where
\begin{enumerate}
 \setlength{\itemsep}{1pt}
 \setlength{\parskip}{0pt}
 \setlength{\parsep}{0pt}
\item $(\vertices,\edges)$ is a directed acyclic graph;
\item $\nlabel\colon\vertices\to\roles_{\exitrole}$ is a labeling function 
assigning roles 
to vertices;
\item $\elabel\colon\edges\to\mesg_{\causal}^{\exit}$ is an edge-labeling 
function that satisfies
\begin{enumerate}
\item $\forall (v,v')\in\edges\colon \elabel(v,v')=\causal\implies \nlabel(v) 
= 
\nlabel(v')$,
\item $\forall (v,v')\in\edges\colon \elabel(v,v')=\exit\implies \nlabel(v') 
= 
\exitrole$;
\end{enumerate}
\item $\ttpfun \colon \mesg^*\to M$ is a function associated with
$\exit$-labeled edges.
\end{enumerate}
\end{definition}
A message edge $(v,v')$ specifies that $\elabel(v,v')=m$
is to be sent from role $r(v)$ to role $r(v')$.
An $\causal$-edge $(v,v')$ represents internal progress of role $\nlabel(v
)=\nlabel(v')$ and allows to specify a causal order in the role's events.
An exit edge denotes that a role can contact the TTP. The TTP then
uses the function $\ttpfun$ to determine a reply to the requesting
role, based on the sequence of messages that it has received. In
Section~\ref{sec:MPCS} exit messages and the $\ttpfun$ function are
used to model the resolve protocol of the TTP.

\begin{figure}
\centering
\includegraphics[scale=.4]{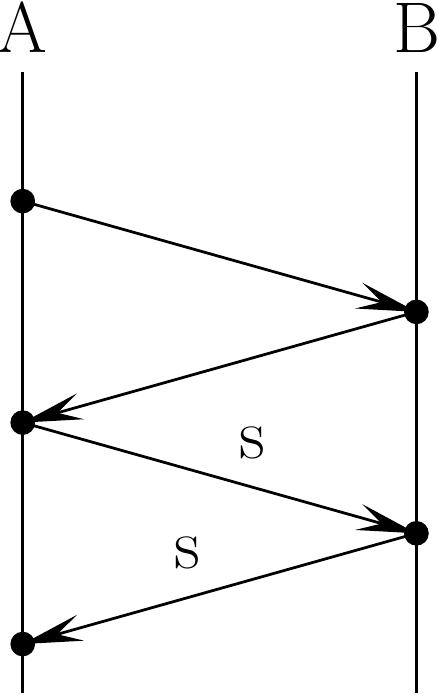}\qquad\includegraphics[scale=.4]{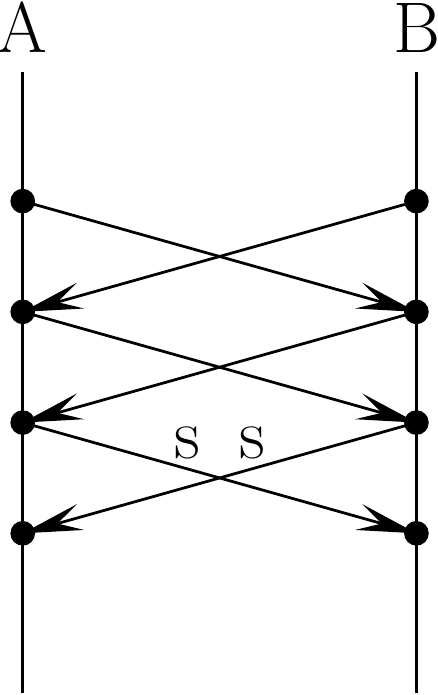}\qquad\includegraphics[scale=.4]{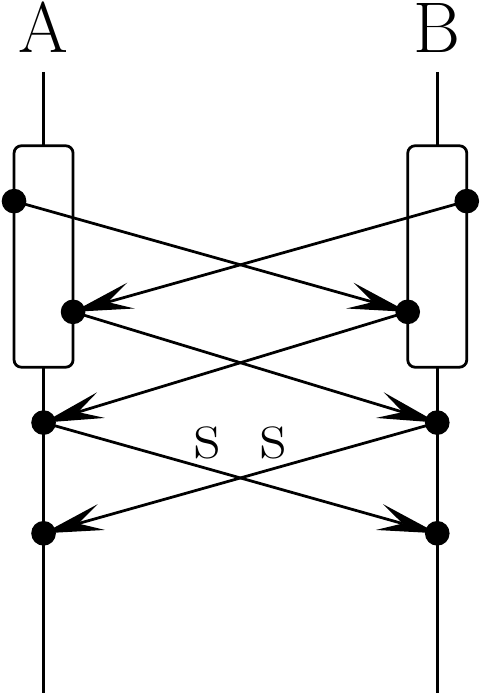}
\caption{Linear, broadcast, and the novel DAG MPCS protocols.}\label{fig:examples}
\end{figure}

We give three examples of DAG protocols in
Figure~\ref{fig:examples}, represented as Message Sequence Charts
(MSCs).
The dots denote the vertices, which we group vertically below their
corresponding role names. Vertical lines in the MSCs correspond to
$\causal$-edges and horizontal or diagonal edges represent message
edges. We mark edges labeled with signing messages with an ``s''
and we leave out the edge labels of promise messages.
We do not display exit edges, 
they are implied by the MPCS protocol
specification. A box represents the splitting of a role into two
parallel threads, which join again at the end of the box.
We revert to a traditional representation of
labeled DAGs
if it is more convenient (see, e.g., Figure~\ref{fig:exampleDAG}).

The first protocol in Figure~\ref{fig:examples} is a classical linear
2-party contract signing
protocol. It consists of one round of promises followed by a round
of exchanging signatures. The second protocol is the classical
broadcast protocol for two signers. It consists of two rounds of
promises, followed by one round of 
signatures.
The third protocol is a novel DAG protocol, showing the power of
in-role parallelism. It is derived from the broadcasting protocol by
observing that its fairness does not depend on the causal order of the
first two vertices of each of the roles.

Let $\protspec=(\vertices,\edges,\nlabel,\elabel,\ttpfun)$ be a protocol specification. 
The \emph{restriction} of $\protspec$ to role $\signer$,
denoted by $\protspec_\signer$, is the protocol specification
$(\vertices_\signer,\edges_\signer,\nlabel_\signer,\elabel_\signer,\ttpfun_\signer)$, where 
\begin{eqnarray*}
\hspace{-4ex}&&\edges_\signer = \set{(v,v')\in\edges\mid
 \nlabel(v)=\signer \vee \nlabel(v')=\signer},\quad \vertices_\signer = \set{v,v'\in\vertices\mid (v,v')\in\edges_\signer},\\
\hspace{-4ex}&&\nlabel_\signer(v) = \nlabel(v)\text{ for }v\in
\vertices_\signer,\quad \elabel_\signer(e) = \elabel(e)\text{ for
}e\in\edges_\signer,\quad \textrm{and}\quad \ttpfun_\signer = \ttpfun.
\end{eqnarray*}

The execution state of a protocol consists of the set of events,
connected to vertices or edges, that have been executed.
\begin{definition}
Let $\protspec=(\vertices,\edges,\nlabel,\elabel,\ttpfun)$ be a
protocol specification.
A \emph{state} of $\protspec$ is a set $\state\subseteq \vertices\cup\edges$.
The set of states of $\protspec$ is denoted by $\statesp$.
The \emph{initial state} of $\protspec$ is defined as
$s_0 = \emptyset$.
\end{definition}

In order to give DAG protocols a semantics,
we first define the \emph{transition relation} between 
states of a protocol.
\begin{definition}
\label{def:transition}
Let $\protspec=(\vertices,\edges,\nlabel,\elabel,\ttpfun)$ 
be a protocol specification, $L=\{\varepsilon,\snd,$ $\rec, \exit\}$ 
the set of transition labels, and
$\state,\state'\in \statesp$ the states of $\protspec$. 
We say that $\protspec$ transitions with label $\alpha$ from state
$\state$ into $\state'$, 
denoted by $\transition{\state}{\alpha}{\state'}$, iff 
$\state\not=\state'$ and one of the following conditions holds
\begin{enumerate}
 \setlength{\itemsep}{1pt}
 \setlength{\parskip}{0pt}
 \setlength{\parsep}{0pt}
\item $\alpha=\rec$ and $\exists v\in\vertices$, such that 
$\inedges(v)\subseteq \state$ and $\state'=\state\cup\{v\}$,
\item $\alpha=\snd$ and $\exists m\in \mesg, e\in\edges$, 
such that 
$\elabel(e)=m$, and
$\state'=\state\cup \{e\}$,
\item $\alpha=\causal$ and $\exists e=(v,v')\in \edges$, such that 
$\elabel(e)=\causal$, 
$v\in\state$ and $\state'=\state\cup \{e\}$,
\item $\alpha=\exit$ and $\exists e\in\edges$, such that 
$\elabel(e)=\exit$ and $\state'=\state\cup \{e\}$.
\end{enumerate}
\end{definition}

In Definition~\ref{def:transition}, receive events are represented by
vertices, all other events by edges. By the first condition, 
a receive event can only occur
if all events assigned to the incoming edges have occurred.
In contrast, the sending of messages (second condition) 
can take place at any time.
The third condition states that an $\causal$-edge can be executed if
the event on which it causally depends has been executed.
Finally, like send events, an exit event can occur at any
time.
Every event may occur at most once, however. This is ensured by the
condition $s'\neq s$.

The transitions model all possible behavior of the system. The
behavior of honest agents in the system will be restricted as detailed
in the following subsection.
We denote sequences by $[a_0,a_1,\ldots,a_l]$ 
and the concatenation of two sequences $\sigma_1$, $\sigma_2$ 
by $\sigma_1\cdot \sigma_2$.

\begin{definition}
\label{def:prot_sem}
Let $\protspec=(\vertices,\edges,\nlabel,\elabel,\ttpfun)$ be a protocol specification
and $L=\{\varepsilon, \snd,$ $\rec, \exit\}$ a set of labels.
The \emph{semantics} of $\protspec$ is the labeled transition system 
$(\statesp,L, {\leadsto,} \state_0)$,
which is a graph consisting of vertices $\statesp$ and edges
$\leadsto$ with start state $\state_0$.
An \emph{execution} of $\protspec$
is a finite sequence 
$\rho=[\state_0,\alpha_1,\state_1,\dots,\alpha_l,\state_l]$,
$l\geq 0$, such that
$\forall i\in\{0,\dots, l-1 \}
\colon
\transition{\state_i}{\alpha_{i+1}}{\state_{i+1}}$. 
The set of all executions of $\protspec$ is denoted by 
$\exep$.
\end{definition} 
If $\rho = [\state_0,\alpha_1,\state_1,\dots,\alpha_l,\state_l]$ is an
execution of $\protspec$ and $\protspec_\signer$ is the restriction to role
$\signer$, 
then the \emph{restricted} execution
$\rho_\signer$ is obtained inductively as follows.
\begin{enumerate}
\item $[\state]_\signer = [\state\cap (V_\signer\cup E_\signer)]$ for a state $\state$.

\item $([\state,\alpha,\state']\cdot \sigma)_\signer = \begin{cases}
[\state]_\signer\cdot \sigma_\signer & \text{ if
} [\state]_\signer = [\state']_\signer\\
[\state]_\signer\cdot[\alpha]\cdot ([\state']\cdot \sigma)_\signer & \text{ else.}
\end{cases}$
\end{enumerate}
Commutativity of restriction and execution is asserted by the
following lemma.

\begin{lemma}
Let $\protspec$ be a protocol specification and $\protspec_\signer$ the
restriction to role $\signer$. Then every restricted execution $\rho_\signer$
is an execution of $\protspec_\signer$.
\end{lemma}

\subsection{Adversary Model}

An honest agent executes the protocol
specification faithfully. The following definition specifies what this
entails for a DAG protocol:
the agent waits for the reception of all causally preceding messages before
sending causally following messages, does not execute an $\exit$ edge
attached to a vertex $v$ if all messages at $v$ have been received and never
executes more than one
$\exit$ edge (which in the context of MPCS protocols corresponds to 
contacting the TTP at most once), 
and does not send any
messages which causally follow a vertex from which the $\exit$ edge was executed. 

\begin{definition}\label{def:honestSigner}
Let $\protspec$ be a DAG protocol specification. 
An agent $\signer$ is \emph{honest} in an execution $\rho$ of
$\protspec$, if all states $\state$ of the restricted 
execution $\rho_\signer$ satisfy the following conditions:
\begin{enumerate}
 \setlength{\itemsep}{1pt}
 \setlength{\parskip}{0pt}
 \setlength{\parsep}{0pt}
\item $\state$ contains at most one exit edge.
\item If $\state$ contains no exit edge, then $\state$ is causally closed.
\item If $\state$ contains an exit edge $e=(v,w)$, $\elabel(e)=\exit$, 
then $v\not\in\state$ and
 $\state\setminus\set{e}$ is causally closed.
\end{enumerate}
\end{definition}

A dishonest agent is only limited by the execution model. 
Thus 
a dishonest agent can send its messages at any time
and in any order, regardless of the causal precedence given in the
protocol specification. A dishonest agent can execute multiple $\exit$
edges and may send and receive messages causally following an exit edge.
Dishonest agents are controlled by a single adversary, their knowledge
is shared with the adversary. The adversary can delay or block messages
sent from one agent to another, but the adversary cannot prevent
messages between agents and the TTP from being delivered eventually.
All communication channels are authentic and confidential. 

\subsection{Communication Complexity}
\label{sec:complexity}

To define measures for expressing the communication 
complexity of DAG protocols, we introduce the notion of 
\emph{closed executions}.
A closed execution is a complete execution of the protocol by honest
agents. 

\begin{definition}
\label{def:exe-closed}
Let $\protspec=(\vertices,\edges,\nlabel,\elabel,\ttpfun)$ 
be a protocol specification and 
$(\statesp,L, {\leadsto,} \state_0)$ be the semantics for $\protspec$.
Given $\rho=[\state_0,\alpha_1,\state_1,\dots,\alpha_l,\state_l]\in \exep$, 
we say that 
$\rho$ is \emph{closed} if the following three conditions are satisfied
\begin{enumerate}
 \setlength{\itemsep}{1pt}
 \setlength{\parskip}{0pt}
 \setlength{\parsep}{0pt}
\item $\state_i$ is causally closed, for every $0\leq i\leq l$,
\item $\alpha_i\not=\exit$, for every $1\leq i\leq l$,
\item $\nexists \alpha\in L\setminus\set{\exit},
\state\in\statesp \colon \rho\cdot[\alpha,\state]\in\exep$.
\end{enumerate}
The set of all closed executions of $\protspec$ is denoted by
$\exepclosed$. 
\end{definition}

Let $\rho=[\state_0,\alpha_1,\state_1,\dots,\alpha_l,\state_l]$ be an execution 
of a protocol $\protspec$. 
By $|\rho|_{\snd}$ we denote the number of labels 
$\alpha_i$, for $1\leq i\leq l$,
such that $\alpha_i=\snd$. 

\begin{lemma}
\label{lem:snd-count}
For any two closed executions $\rho$ and $\rho'$ of a protocol 
$\protspec$ we have $|\rho|_{\snd}=|\rho'|_{\snd}$.
\end{lemma}
The proof is given in the appendix.
The first measure expressing the complexity of a protocol $\protspec$
is called \emph{message complexity}.
It counts the overall number of messages that have been sent 
in a closed execution of a protocol $\protspec$.
\begin{definition}
Let $\protspec$ be a protocol specification and let
$\rho\in\exepclosed$.
The \emph{message complexity} of $\protspec$, denoted by 
$MC_\protspec$, is defined as
$MC_\protspec = |\rho|_{\snd}$.
\end{definition}
Lemma~\ref{lem:snd-count} guarantees that the message complexity of 
a protocol is well defined.

The second complexity measure 
is called \emph{parallel complexity}. 
It represents the 
minimal time of a closed execution 
assuming that all events which can be executed in parallel 
are executed in parallel.
The parallel complexity of a protocol is defined as the length of a
maximal chain of causally related send edges.

\begin{definition}
The \emph{parallel complexity} of a protocol $\protspec$, 
denoted by $PC_\protspec$, 
is defined as 
$$
PC_\protspec=\max_{n \in \naturals}\,
\exists_{[e_{1},e_{2},\ldots,e_{n}] \in \edges^{*}} :\\
\forall_{1\leq i \leq n}: \elabel(e_{i}) = \snd \wedge
\forall_{1\leq i < n}: e_{i}\prec e_{i+1}\text{.}
$$
\end{definition}

\begin{example}
The message complexity of the first protocol in
Figure~\ref{fig:examples} is 4, which is known to be optimal for two
signers~\cite{schunterphd}. Its parallel complexity is 4, too. 
The message complexity of the other two protocols is 6, but their
parallel complexity is 3, which is optimal 
for broadcasting protocols with two signers~\cite{bw00}.
\end{example}

\section{DAG MPCS protocols}\label{sec:MPCS}

We now define a class of optimistic MPCS protocols in the DAG protocol model.

\subsection{Main Protocol}\label{sec:mainprotocol}

The key requirements we want our DAG MPCS protocol specification
to satisfy, stated formally in Definition~\ref{def:parallelMPCSspec}, are as follows.
The messages exchanged between signers in the protocol are of two types,
promises, denoted by $\promises()$, and signatures, denoted by $\signature()$. 
Every promise contains information about 
the vertex from which it is sent. This is done by concatenating
the contract $\contract$ with the vertex $v$ the promise originates from and is denoted by $(\contract,v)$.
The signers can contact the TTP at any time.
This is modeled with exit edges:
Every vertex $v\in\vertices$ such that $\nlabel(v)\in\signers$ 
(the set of all signers) is adjacent 
to a unique vertex $v_\ttp\in\vertices$, $\nlabel(v_\ttp)=\exitrole$. 
The communication with $\ttp$ is represented by $\ttpfun$.
The set of vertices with outgoing signature messages is denoted by
$\Sigset$.

\begin{definition}\label{def:parallelMPCSspec}
Let $\protspec=(\vertices,\edges,\nlabel,\elabel,\ttpfun)$ be a protocol
specification, $\signers\subset\roles$ be a finite set of signers, 
$\contract$ be a contract, and $\Sigset\subseteq\vertices$.
$\protspec$ is called a \emph{DAG MPCS protocol specification for
 $\contract$,} if\ \footnote{We write $\exists!$ for unique existential
 quantification.}
\begin{enumerate}
\item\label{cond:TTP} $\exists!\ v_\ttp\in\vertices: \nlabel(v_\ttp)=\exitrole\wedge
\forall v\in\vertices\setminus\set{v_\ttp}: (v,v_\ttp)\in\edges$,
\item\label{cond:transitivity} $\forall v,w\in\vertices: v\prec w\Rightarrow (v,w)\in\edges\vee
 \exists u\in\vertices: v\prec u\prec w\wedge \nlabel(u)\in\set{\nlabel(v),\nlabel(w)}$,

\item\label{cond:elabel} $\forall (v,w)\in\edges: \elabel(v,w) = $
$$\begin{cases}
\causal\text{, \quad if } \nlabel(v)=\nlabel(w)\text,\\
\exit\text{, \quad if } w=v_\ttp\text,\\
\signature_{\nlabel(v)}(\contract)\text{, \quad if }
v\in\Sigset\wedge\nlabel(v)\neq \nlabel(w)\neq\exitrole\text,\\
\promise{\nlabel(v)}{\contract}{v}{\nlabel(w)}{\ttp}\text{, \quad else.}
\end{cases}
$$
\item $\ttpfun:
 \mesg^*\to\set{\abort,\displaystyle\left(\signature_\signer(\contract)\right)_{\signer\in\signers}}$,
 where $(\signature_\signer(\contract))_{\signer\in\signers}$ denotes a list of
 signatures on $\contract$, one by each signer.
\end{enumerate}

We write $\Sigset(\protspec)$ for the largest subset of $\Sigset$ which satisfies
$$v\in\Sigset(\protspec)\Rightarrow \exists w\in\vertices:
 (v,w)\in\edges, \elabel(v,w)\in\mesg.$$
The set $\Sigset(\protspec)$ is called the \emph{signing set}.
\end{definition}

We represent DAG MPCS protocols as \emph{skeletal} graphs as shown in
Figure~\ref{fig:skeletal}. 
The full graph, shown in Figure~\ref{fig:full}, is obtained from the
skeletal graph by adding all edges required by 
condition~\ref{cond:transitivity} of
Definition~\ref{def:parallelMPCSspec} and extending $\mu$ according to
condition~\ref{cond:elabel}. 
The $\causal$ edges are dashed in the graphs. 
The shaded vertices in the graphs indicate the vertices that are
in the signing set. 
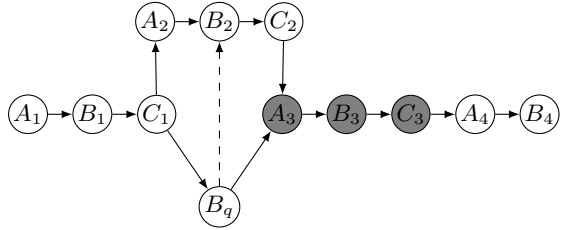
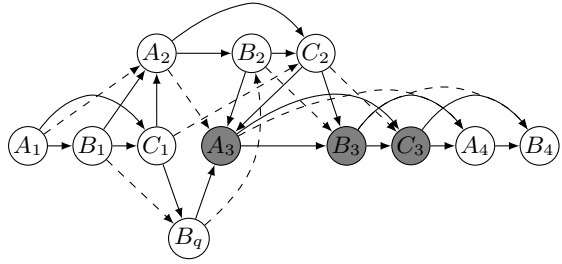
\begin{figure}
 \centering
\subfloat[\label{fig:skeletal}Skeletal graph.]{\begin{tikzpicture}[>=latex,line join=bevel]
\begin{scope}
  \pgfsetstrokecolor{black}
  \definecolor{strokecol}{rgb}{1.0,1.0,1.0};
  \pgfsetstrokecolor{strokecol}
  \definecolor{fillcol}{rgb}{1.0,1.0,1.0};
  \pgfsetfillcolor{fillcol}
  \filldraw (0bp,-1bp) -- (0bp,138bp) -- (341bp,138bp) -- (341bp,-1bp) -- cycle;
\end{scope}
\begin{scope}
  \pgfsetstrokecolor{black}
  \definecolor{strokecol}{rgb}{1.0,1.0,1.0};
  \pgfsetstrokecolor{strokecol}
  \definecolor{fillcol}{rgb}{1.0,1.0,1.0};
  \pgfsetfillcolor{fillcol}
  \filldraw (0bp,-1bp) -- (0bp,138bp) -- (341bp,138bp) -- (341bp,-1bp) -- cycle;
\end{scope}
  \node (A_4) at (289bp,69bp) [draw,circle] {$A_4$};
  \node (A_3) at (169bp,69bp) [draw,fill=gray,circle] {$A_3$};
  \node (A_2) at (90bp,127bp) [draw,circle] {$A_2$};
  \node (A_1) at (11bp,69bp) [draw,circle] {$A_1$};
  \node (C_1) at (91bp,69bp) [draw,circle] {$C_1$};
  \node (C_3) at (249bp,69bp) [draw,fill=gray,circle] {$C_3$};
  \node (C_2) at (170bp,127bp) [draw,circle] {$C_2$};
  \node (B_q) at (130bp,11bp) [draw,circle] {$B_q$};
  \node (B_1) at (51bp,69bp) [draw,circle] {$B_1$};
  \node (B_2) at (130bp,127bp) [draw,circle] {$B_2$};
  \node (B_3) at (209bp,69bp) [draw,fill=gray,circle] {$B_3$};
  \node (B_4) at (329bp,69bp) [draw,circle] {$B_4$};
  \draw [->] (B_1) ..controls (67.177bp,69bp) and (72.105bp,69bp)  .. (C_1);
  \draw [->] (C_3) ..controls (265.18bp,69bp) and (270.1bp,69bp)  .. (A_4);
  \draw [->] (C_1) ..controls (90.65bp,89.6bp) and (90.405bp,103.31bp)  .. (A_2);
  \draw [->] (B_q) ..controls (142.76bp,30.322bp) and (154.02bp,46.485bp)  .. (A_3);
  \draw [->] (B_3) ..controls (225.18bp,69bp) and (230.1bp,69bp)  .. (C_3);
  \draw [->] (C_2) ..controls (169.66bp,106.79bp) and (169.41bp,93.089bp)  .. (A_3);
  \draw [->] (A_3) ..controls (185.18bp,69bp) and (190.1bp,69bp)  .. (B_3);
  \draw [->,dashed] (B_q) ..controls (130bp,43.114bp) and (130bp,90.415bp)  .. (B_2);
  \draw [->] (A_1) ..controls (27.177bp,69bp) and (32.105bp,69bp)  .. (B_1);
  \draw [->] (A_2) ..controls (106.18bp,127bp) and (111.1bp,127bp)  .. (B_2);
  \draw [->] (C_1) ..controls (103.95bp,49.412bp) and (115.08bp,33.429bp)  .. (B_q);
  \draw [->] (A_4) ..controls (305.18bp,69bp) and (310.1bp,69bp)  .. (B_4);
  \draw [->] (B_2) ..controls (146.18bp,127bp) and (151.1bp,127bp)  .. (C_2);
\end{tikzpicture}

}\\
\subfloat[\label{fig:full}Full graph.]{\begin{tikzpicture}[>=latex,line join=bevel]
\begin{scope}
  \pgfsetstrokecolor{black}
  \definecolor{strokecol}{rgb}{1.0,1.0,1.0};
  \pgfsetstrokecolor{strokecol}
  \definecolor{fillcol}{rgb}{1.0,1.0,1.0};
  \pgfsetfillcolor{fillcol}
  \filldraw (0bp,-1bp) -- (0bp,159bp) -- (341bp,159bp) -- (341bp,-1bp) -- cycle;
\end{scope}
\begin{scope}
  \pgfsetstrokecolor{black}
  \definecolor{strokecol}{rgb}{1.0,1.0,1.0};
  \pgfsetstrokecolor{strokecol}
  \definecolor{fillcol}{rgb}{1.0,1.0,1.0};
  \pgfsetfillcolor{fillcol}
  \filldraw (0bp,-1bp) -- (0bp,159bp) -- (341bp,159bp) -- (341bp,-1bp) -- cycle;
\end{scope}
  \node (A_4) at (289bp,69bp) [draw,circle] {$A_4$};
  \node (A_3) at (131bp,69bp) [draw,fill=gray,circle] {$A_3$};
  \node (A_2) at (91bp,127bp) [draw,circle] {$A_2$};
  \node (A_1) at (11bp,69bp) [draw,circle] {$A_1$};
  \node (C_1) at (91bp,69bp) [draw,circle] {$C_1$};
  \node (C_3) at (249bp,69bp) [draw,fill=gray,circle] {$C_3$};
  \node (C_2) at (190bp,127bp) [draw,circle] {$C_2$};
  \node (B_q) at (111bp,11bp) [draw,circle] {$B_q$};
  \node (B_1) at (51bp,69bp) [draw,circle] {$B_1$};
  \node (B_2) at (150bp,127bp) [draw,circle] {$B_2$};
  \node (B_3) at (209bp,69bp) [draw,fill=gray,circle] {$B_3$};
  \node (B_4) at (329bp,69bp) [draw,circle] {$B_4$};
  \draw [->,dashed] (A_2) ..controls (104.34bp,107.33bp) and (115.89bp,91.16bp)  .. (A_3);
  \draw [->] (C_2) ..controls (171.54bp,108.48bp) and (152.07bp,90.002bp)  .. (A_3);
  \draw [->] (B_2) ..controls (143.48bp,106.77bp) and (138.59bp,92.357bp)  .. (A_3);
  \draw [->,dashed] (C_1) ..controls (118.65bp,85.64bp) and (158.8bp,108.35bp)  .. (C_2);
  \draw [->,dashed] (A_1) ..controls (34.465bp,86.426bp) and (64.164bp,107.21bp)  .. (A_2);
  \draw [->,dashed] (C_2) ..controls (208.46bp,108.48bp) and (227.93bp,90.002bp)  .. (C_3);
  \draw [->] (C_3) ..controls (265.18bp,69bp) and (270.1bp,69bp)  .. (A_4);
  \draw [->,dashed] (B_2) ..controls (168.46bp,108.48bp) and (187.93bp,90.002bp)  .. (B_3);
  \draw [->] (B_q) ..controls (117.88bp,31.253bp) and (123.06bp,45.769bp)  .. (A_3);
  \draw [->] (B_3) ..controls (225.18bp,69bp) and (230.1bp,69bp)  .. (C_3);
  \draw [->,dashed] (B_q) ..controls (129.13bp,26.393bp) and (144.7bp,41.148bp)  .. (151bp,58bp) .. controls (157.82bp,76.247bp) and (155.44bp,99.221bp)  .. (B_2);
  \draw [->,dashed] (A_3) ..controls (153.21bp,81.737bp) and (176.4bp,93.527bp)  .. (198bp,98bp) .. controls (224.98bp,103.59bp) and (234.54bp,108.55bp)  .. (260bp,98bp) .. controls (268.7bp,94.395bp) and (276.16bp,86.769bp)  .. (A_4);
  \draw [->] (A_3) ..controls (155.9bp,87.889bp) and (191.32bp,109.88bp)  .. (220bp,98bp) .. controls (228.7bp,94.395bp) and (236.16bp,86.769bp)  .. (C_3);
  \draw [->] (C_1) ..controls (97.763bp,49.065bp) and (102.94bp,34.556bp)  .. (B_q);
  \draw [->] (B_2) ..controls (166.18bp,127bp) and (171.1bp,127bp)  .. (C_2);
  \draw [->] (A_1) ..controls (22.314bp,85.149bp) and (30.396bp,94.022bp)  .. (40bp,98bp) .. controls (49.033bp,101.74bp) and (52.967bp,101.74bp)  .. (62bp,98bp) .. controls (70.704bp,94.395bp) and (78.157bp,86.769bp)  .. (C_1);
  \draw [->] (B_1) ..controls (67.177bp,69bp) and (72.105bp,69bp)  .. (C_1);
  \draw [->] (B_1) ..controls (64.087bp,88.322bp) and (75.632bp,104.48bp)  .. (A_2);
  \draw [->,dashed] (B_1) ..controls (69.774bp,50.477bp) and (89.57bp,32.002bp)  .. (B_q);
  \draw [->] (A_3) ..controls (159.91bp,69bp) and (177.55bp,69bp)  .. (B_3);
  \draw [->] (A_2) ..controls (108.62bp,141.28bp) and (123.81bp,151.77bp)  .. (139bp,156bp) .. controls (148.42bp,158.62bp) and (151.97bp,159.74bp)  .. (161bp,156bp) .. controls (169.7bp,152.39bp) and (177.16bp,144.77bp)  .. (C_2);
  \draw [->] (A_2) ..controls (113.44bp,127bp) and (124.81bp,127bp)  .. (B_2);
  \draw [->] (C_1) ..controls (91bp,89.6bp) and (91bp,103.31bp)  .. (A_2);
  \draw [->] (C_2) ..controls (196.52bp,106.77bp) and (201.41bp,92.357bp)  .. (B_3);
  \draw [->] (B_3) ..controls (220.31bp,85.149bp) and (228.4bp,94.022bp)  .. (238bp,98bp) .. controls (247.03bp,101.74bp) and (250.97bp,101.74bp)  .. (260bp,98bp) .. controls (268.7bp,94.395bp) and (276.16bp,86.769bp)  .. (A_4);
  \draw [->] (A_1) ..controls (27.177bp,69bp) and (32.105bp,69bp)  .. (B_1);
  \draw [->] (A_4) ..controls (305.18bp,69bp) and (310.1bp,69bp)  .. (B_4);
  \draw [->,dashed] (B_3) ..controls (220.31bp,85.149bp) and (228.4bp,94.022bp)  .. (238bp,98bp) .. controls (263.46bp,108.55bp) and (274.54bp,108.55bp)  .. (300bp,98bp) .. controls (308.7bp,94.395bp) and (316.16bp,86.769bp)  .. (B_4);
  \draw [->] (C_3) ..controls (260.31bp,85.149bp) and (268.4bp,94.022bp)  .. (278bp,98bp) .. controls (287.03bp,101.74bp) and (290.97bp,101.74bp)  .. (300bp,98bp) .. controls (308.7bp,94.395bp) and (316.16bp,86.769bp)  .. (B_4);
\end{tikzpicture}

}
 \caption{Skeletal and full representation of a DAG MPCS protocol.}
\label{fig:exampleDAG}
\end{figure}
We define the \emph{knowledge} $\knowledge(v)$ 
of a vertex $v$ to be the set of 
message edges causally preceding $v$, and incoming to a vertex
of the same role. The knowledge of a vertex represents the state right
after its receive event.
\begin{eqnarray*}
\knowledge(v)& =& \{(w,v')\in\edges\mid \elabel(w,v')\in\mesg,
v'\preceq v, \nlabel(v')=\nlabel(v)\}
\end{eqnarray*}

We define the \emph{pre-knowledge} $\preknowledge(v)$ of a vertex $v$ to be
$\preknowledge(v) = \{(w,v')\in\knowledge(v)\mid v'\prec v\}$.
The pre-knowledge represents the state just 
\emph{before} the vertex' receive event has taken place.
We extend both definitions to sets $S\subseteq\vertices$:
$$
\knowledge(S) = \bigcup_{v\in S} \knowledge(v)
\quad\text{ and }\quad
\preknowledge(S) = \bigcup_{v\in S} \preknowledge(v)\text.
$$
We define the \emph{initial set} of $\protspec$, denoted $\Initset(\protspec)$ 
to be the set of vertices
of the protocol specification for which the pre-knowledge of 
the same role does not contain
an incoming edge by every other role. Formally,
\begin{eqnarray*}
v\in \Initset(\protspec)&\Longleftrightarrow&\set{\nlabel(w)\in\signers\mid (w,v')\in \preknowledge(v)} \cup\set{\nlabel(v)}\neq\signers
\end{eqnarray*}
The \emph{end set} of $\protspec$, denoted $\Endset(\protspec)$, is
the set of vertices of the protocol specification at which the
corresponding signer 
possesses all signatures.
\begin{eqnarray*}
v\in \Endset(\protspec)&\Longleftrightarrow&\set{\nlabel(w)\in\signers\mid (w,v')\in \knowledge(v), w\in\Sigset(\protspec)} \cup\set{\nlabel(v)}=\signers
\end{eqnarray*}

\subsection{Resolve Protocol}\label{sec:resolve}

Let $\protspec=(\vertices,\edges,\nlabel,\elabel,\ttpfun)$
be a DAG MPCS protocol specification.
The resolve protocol is a two-message protocol between a signer and
the TTP $\ttp$, initiated by the signer. The communication channels for this
protocol are assumed to be resilient, confidential, and authentic. $\ttp$
 is assumed to respond immediately to the signer. This is modeled
in $\protspec$ via an exit edge from a vertex
$v\in\vertices\setminus\set{v_\ttp}$ 
to the unique vertex $v_\ttp\in\vertices$. $\ttp$'s response is
given by the $\ttpfun$ function,
$\ttpfun:\mesg^*\to\set{\abort,(\signature_\signer(\contract))_{\signer\in\signers}}$.
If $m_1,\ldots,m_n$ is the sequence of messages sent by the signers 
to $\ttp$, then $\ttpfun(m_1,\ldots,m_n)$ is $\ttp$'s response for the
last signer in the sequence. The function will 
be stated formally in Definition~\ref{def:ttpfun}. 

We denote the resolve protocol
in the following by $\res(\contract,v)$.
{} 
The signer initiating $\res(\contract,v)$ is $\nlabel(v)$. 
He sends the list of messages assigned to the vertices in his pre-knowledge
$\preknowledge(v)$, prepended
by $\promise{\nlabel(v)}{\contract}{v}{\nlabel(v)}{\ttp}$,
to $\ttp$. This demonstrates that $\nlabel(v)$ has executed
all receive events causally preceding $v$. We denote $\nlabel(v)$'s
message for $\ttp$ by $p_v$:
\begin{equation}\label{eq:PCS}
p_v = \big(\promise{\nlabel(v)}{\contract}{v}{\nlabel(v)}{\ttp}, (\elabel(w,v'))_{(w,v')\in \preknowledge(v)}\big)
\end{equation}
The promise 
$\promise{\nlabel(v)}{\contract}{v}{\nlabel(v)}{\ttp}$, which is the first
element of $p_v$, is used by $\ttp$ to extract the contract
$\contract$, to learn at which step in the protocol 
$\nlabel(v)$ claims to be, and to create a 
signature on behalf of $\nlabel(v)$ when necessary.
All messages received from the signers are stored. $\ttp$ performs a
deterministic decision procedure, shown in
Algorithm~\ref{algo:resolve}, on the received message and existing
stored messages and \emph{immediately} sends back $\abort$ or the list 
of signatures $(\signature_\signer(\contract))_{\signer\in\signers}$.

Our decision procedure 
is based on~\cite{mr08,KR12}.
The input to the algorithm consists of a message $m$ received by the
$\ttp$ from a signer and state information which is maintained by $\ttp$.
$\ttp$
extracts the contract and the DAG MPCS protocol specification
from $m$.
For each contract $\contract$, $\ttp$ maintains the following state
information. A list
$\evidence_\contract$ of
all messages received from signers, a set $I_\contract$
of vertices the signers contacted $\ttp$ from, 
a set $\dishonest_\contract$ of signers
considered to be dishonest, and the last decision made $\decision_\contract$.
If $\ttp$ has not been contacted by any
signer regarding contract $\contract$, then
$\decision_\contract=\abort$. Else, $\decision_\contract$ is equal to
$\abort$ or the list $(\signature_\othersigner(\contract)_{\othersigner\in\signers})$
of 
signatures on $\contract$, one by each signer.

$\ttp$ verifies that the request is legitimate 
in that the received message $m$ is valid and the requesting signer
$\signer$ is
not already considered to be dishonest. 
If these preliminary checks pass, the message is appended to $\evidence_\contract$.
This is described in Algorithm~\ref{algo:resolve} in
lines~\ref{line:Rfirst} through~\ref{line:evidence}.
The main part of the algorithm, starting at
line~\ref{line:forcedAbort}, concerns the detection of signers who have
continued the main protocol execution after executing the resolve protocol. 
If $\signer$ has not received a promise from every other signer in the
protocol (i.e.~the if clause in line~\ref{line:forcedAbort} is not satisfied),
then $\ttp$ sends back the last decision made (line~\ref{line:retDecision}). This decision is an
\abort{} token unless $\ttp$ has been contacted before and decided to send back a signed contract.
If $\signer$ has received a promise from every other signer, 
$\ttp$ computes the set of dishonest signers
(lines~\ref{line:caughtRepeat1} through~\ref{line:caughtRepeat3}) by adding to it
every signer which has carried out the resolve protocol, but can be
seen to have continued the protocol execution
(line~\ref{line:caughtRepeat2}) based on the evidence the TTP has collected.
If $\signer$ is the only honest signer that has
contacted $\ttp$ until this point in time, the decision is made to henceforth
return a signed contract. 

\LinesNumbered
\begin{algorithm}
\SetKwInOut{Input}{input}\SetKwInOut{Output}{output}
\SetKw{Return}{return}
\SetKw{Extract}{extract}
\SetKw{Result}{output}
\Input{$m,r,\decision_\contract,\evidence_\contract,I_\contract,\dishonest_\contract$}
\Output{$r,\decision_\contract,\evidence_\contract,I_\contract,\dishonest_\contract$}
\BlankLine
\If{\nllabel{line:Rfirst}$m \neq (\promise{\signer}{\contract}{v}{\signer}{\ttp},\mathit{list})$}
{$r=\abort$\; \Return{\Result}\;}
\If{\nllabel{line:verifRepeat}$\signer\in\dishonest_\contract
\vee \forall w\in\vertices: m \neq p_w 
\vee \exists w'\in I_\contract: \signer=\nlabel(w')$}
{$\dishonest_\contract:=\dishonest_\contract\cup\set{\signer}$\;$r=\abort$\;\Return{\Result}\;}
$I_\contract:=I_\contract\cup \{v\}$\;\nllabel{line:addToI}
$\evidence_\contract:=(\evidence_\contract, m)$\;\nllabel{line:evidence}
\If{$v\notin\Initset(\protspec)$\nllabel{line:forcedAbort}}
{\For{$w\in I_\contract$\nllabel{line:caughtRepeat1}}
{\If{\nllabel{line:caughtRepeat2}$w\prec (w',x)\in \preknowledge(I_\contract)\wedge \nlabel(w')=\nlabel(w)$}
{$\dishonest_\contract:=\dishonest_\contract\cup\{\nlabel(w)\}$\;\nllabel{line:caughtRepeat3}}}
\If{\nllabel{line:P_honest}$\forall w\in I_\contract: \nlabel(w)\notin\dishonest_\contract\implies \nlabel(w)= \signer$}
{$\decision_\contract:=(\signature_\othersigner(\contract))_{\othersigner\in \signers}$\;\nllabel{line:resolve}
}}
$r=\decision_\contract$\;
\Return{\Result}\;\nllabel{line:retDecision} 
\caption{TTP decision procedure $\ttpfun_0$}
\label{algo:resolve}
\end{algorithm}

\begin{definition}\label{def:ttpfun}
Let $\protspec=(\vertices,\edges,\nlabel,\elabel,\ttpfun)$ 
be a DAG MPCS protocol specification and $\ttpfun_0$ the TTP
decision procedure from Algorithm~\ref{algo:resolve}.
Then $\ttpfun: \mesg^*\to\mesg$ is defined for
$m_1,\ldots,m_n\in\mesg$ by
$$\ttpfun(m_1,\ldots,m_n) = \pi_1(\ttpfun_1(m_1,\ldots,m_n))\text,$$
where 
$\pi_1$ is the projection to the first coordinate and $\ttpfun_1$ is
defined inductively by
\begin{eqnarray*}
\ttpfun_1() &=& (\abort,\abort,\emptyset,\emptyset,\emptyset)\\
\ttpfun_1(m_1,\ldots,m_n) &=& \ttpfun_0(m_n,\ttpfun_1(m_1,\ldots,m_{n-1}))\text.
\end{eqnarray*}
\end{definition}

Thus the $\ttpfun$ function represents the response of the TTP in the
$\res(\contract,v)$ protocol for all executions of
$\protspec$. 

\subsection{Fairness}\label{sec:fairness}

We say that a DAG MPCS protocol execution is fair for
signer $\signer$ if one of the following three conditions is true:
(i) No signer has received a signature of $\signer$; 
(ii) $\signer$ has received signatures of all other signers; 
(iii) $\signer$ has not received an $\abort$ token from the TTP.
The last condition allows an execution to be fair as long as there is
a possibility for the 
signer to receive signatures of all other signers.

The key problem in formalizing these conditions is to
capture under which circumstances the TTP responds with an
$\abort$ token to a request by a signer. The TTP's response is dependent
on the decision procedure which in turn depends on the order in which
the TTP is contacted by the signers. Since the decision procedure is
deterministic, it follows that the $\ttpfun$ function can be
determined for every execution $\rho =
[\state_0,\alpha_1,\state_1,\dots,\state_n]$ by considering the
pre-knowledge of the vertices from which the $\exit$ transitions are
taken. Abusing notation, we will write $\ttpfun(\rho)$ instead of
$\ttpfun(m_1,\ldots,m_k)$ where $m_i$ are the messages sent to the TTP
at the $i$-th $\exit$ transition in the execution.

\begin{definition}\label{def:fair-execution}
Let $\ttp$ be the TTP.
An execution $\rho=[\state_0,\alpha_1,\dots,\state_n]$ of $\protspec$ 
is \emph{fair} for
signer $\signer$ if one of the following conditions is satisfied:
\begin{enumerate}
\item 
$\signer$ has not sent a signature and no
 signer has received signatures from $\ttp$.
$$\ttpfun(\rho)=\abort\wedge\forall (v,w)\in\state_n: \nlabel(v)=\signer, \nlabel(w)\neq \signer \Longrightarrow v\not\in\Sigset(\protspec)$$
\item 
$\signer$ has received signatures from all other signers.
$$\exists v\in\state\cap\Endset(\protspec):
 \nlabel(v)=\signer$$
\item 
$\signer$ has not received an $\abort$ token from $\ttp$.
$$
\exists (v,w)\in\state: \nlabel(v)=\signer\wedge \nlabel(w)=\exitrole
\Rightarrow \ttpfun([\state_0,\dots,\state_{k},\exit,\state_k\cup\set{(v,w)}])\neq\abort
$$
\end{enumerate}
If none of these conditions are satisfied, 
the execution is unfair for $\signer$.
\end{definition}

\begin{definition}\label{def:fairness}
 An MPCS protocol specification $\protspec$ is said to be
 fair, if every 
 execution $\rho$ of $\protspec$ is fair for all signers
 that are honest in $\rho$.
\end{definition}

\subsection{Sufficient and necessary conditions}\label{sec:suffnec}
By the TTP decision procedure, $\ttp$ returns an
\abort{} token if contacted from a vertex 
$v\in\Initset(\protspec)$. Thus a necessary condition for fairness is that an
honest signer executes all steps of the initial set
causally before all steps of the signing set that are not in the end set:
\begin{eqnarray}\label{condition1}
\hspace{-4ex}&&\forall v\in\Initset(\protspec), w\in\Sigset(\protspec)\setminus\Endset(\protspec): \nlabel(v)=\nlabel(w)\implies v\prec w
\end{eqnarray}
If $\signer$ contacts $\ttp$ from a vertex $v\not\in\Initset(\protspec)$, then 
$\ttp$ responds with an \abort{} token if it has already issued an
\abort{} token to another signer who is not in the set $\dishonest_\contract$.
This condition can be exploited by a group of dishonest
signers in an \emph{abort chaining attack}~\cite{MR07}. The
following definition states the requirements for a successful abort
chaining attack.
For ease of reading, we define the predicate $\hon(v,I)$. The
predicate is true if there is no evidence (pre-knowledge) at the
vertices in $I$ that the signer $\nlabel(v)$ has sent a message
at or causally after $v$:
$$\hon(v,I) \equiv \neg\exists{(x,y)\in \preknowledge(I)}:
v\prec(x,y)\wedge \nlabel(v)=\nlabel(x)$$
This is precisely 
the criterion used by $\ttp$ to verify
honesty
in Algorithm~\ref{algo:resolve}, line~\ref{line:caughtRepeat2}. 

\begin{definition}\label{def:PACseq}
Let $\contract$ be a contract and $l\leq\abs{\signers}$.
A sequence $(v_1, \ldots, v_l\mid s)$ over $\vertices$ is called an
abort-chaining sequence (AC sequence) for $\protspec$
if the following conditions hold:
 \begin{enumerate}
 \item\label{PACcond:forcedAbort} Signer $\nlabel(v_1)$ receives an
 abort token: $v_1\in\Initset(\protspec)$
 \item\label{PACcondition1} No signer contacts $\ttp$ more than once:
 $\forall_{i \neq j} \; \nlabel(v_i) \neq \nlabel(v_j)$
 \item\label{PACcondition4} The present and previous signer to
 contact $\ttp$ are considered honest by $\ttp$:
$$\forall i\leq l:\hon(v_i,\set{v_1,\ldots,v_i})\wedge \hon(v_{i-1},\set{v_1,\ldots,v_i})$$

\item\label{PACconditionNew} The last signer to contact $\ttp$ has
 not previously received all signatures:
$$\forall v\prec v_l: \nlabel(v)=\nlabel(v_l)\implies v\not\in\Endset(\protspec)$$

\item\label{PACcondition3} The last signer to contact $\ttp$ has sent
 a signature before contacting $\ttp$ or in a parallel thread:
$$s\in\Sigset(\protspec)\setminus\Endset(\protspec): \nlabel(s)=\nlabel(v_l) \wedge v_l\not\preceq s$$
 \end{enumerate}
\end{definition}

The AC sequence represents the order in which signers execute the
resolve protocol with $\ttp$. A vertex $v_i$ in the sequence implies an
exit transition via the edge $(v_i,v_\ttp)$ in the protocol execution.
An abort
chaining attack must start at a step in which $\ttp$ has no choice
but to respond with an abort token due to lack of information. 
Condition~\ref{PACcond:forcedAbort} covers this. 
Each signer may run the resolve protocol at most
once. This is covered by Condition~\ref{PACcondition1}. 
To ensure that $\ttp$ continues to issue \abort{} tokens, 
Condition~\ref{PACcondition4}
requires that there must always be a signer
which 
according to $\ttp$'s evidence has not 
continued protocol execution after contacting $\ttp$.
To complete an abort chaining attack, there
needs to be a signer which has issued a signature
(Condition~\ref{PACcondition3}), but has not received a signature
(Conditions~\ref{PACconditionNew} and~\ref{PACcondition3}) and 
will not receive a signed contract from $\ttp$ because there is an
honest signer (by Condition~\ref{PACcondition4}) which has received an
$\abort$ token. 

It is not surprising (but nevertheless proven in the appendix)
that a
protocol with an AC sequence is unfair. However, the converse is true, too.
\begin{theorem}\label{th:fairness}
Let $\protspec$ be a DAG MPCS protocol.
Then $\protspec$ is fair if and only if it has no AC sequences.
\end{theorem}

The proof of this and the following theorems is given in the appendix.

\subsection{Fairness criteria}\label{sec:criteria}

Theorem~\ref{th:fairness} reduces 
the verification of fairness from 
analyzing all
executions to verifying that there is no AC-sequence (Definition~\ref{def:PACseq}).
This, however, is still
difficult to verify in general. The following two 
results are tools
to quickly assess fairness of DAG MPCS protocols. The first is an
unfairness criterion and the second is a fairness criterion for a
large class of DAG MPCS protocols.

The following theorem states that
in a fair DAG MPCS protocol, 
the union of paths from the initial set to every vertex
$v\in\Sigset(\protspec)$ 
must contain all permutations of all signers (other than
$\nlabel(v)$) as subsequences. In the class of linear MPCS
protocols, considered in~\cite{KR12}, this criterion was both
necessary and sufficient. We show in Example~\ref{ex:insufficient} below that this
criterion is not sufficient for fairness of DAG MPCS protocols. 

For $I\subseteq\vertices$, $v\in\vertices$, we denote by 
$\paths(I,v)=\{(v_1,\ldots,v_k)\in\vertices^*\mid v_1\in I, v_k=v,
 \forall_{1\leq i< k}: (v_i,v_{i+1})\in\edges\}$ the set of all
 directed paths from a vertex in $I$ to $v$.
If $p=(v_1,\ldots,v_k)$ is a sequence of vertices, 
we denote by $\nlabel(p)=(\nlabel(v_1),\ldots,\nlabel(v_k))$ the corresponding sequence of signers.
The sequences of signers corresponding to the paths from $I$ to $v$
is denoted by $\seq(I,v) = \set{\nlabel(p)\in \signers^*\mid
 p\in\paths(I,v)}$.

\begin{theorem}\label{thm:unfairness}
Let $k=\abs{\signers}$. Let $\protspec$ be an optimistic fair DAG MPCS protocol, 
$$I=\set{v\in\Initset(\protspec)\mid (v\prec w\wedge
 \nlabel(v)=\nlabel(w))\Rightarrow w\not\in\Initset(\protspec)}.$$
If $v\in\Sigset(\protspec)$, then 
for every permutation $(\signer_1,\ldots,\signer_{k-1})$ of
signers in $\signers\setminus\set{\nlabel(v)}$, there exists a sequence in 
$\seq(I,v)$ which contains $(\signer_1,\ldots,\signer_{k-1})$ as a (not
necessarily consecutive) subsequence.
\end{theorem}

The converse of the theorem is not true as the following example
shows. In particular, this example demonstrates that the addition of a
vertex to a fair DAG MPCS protocol does not necessarily preserve
fairness.

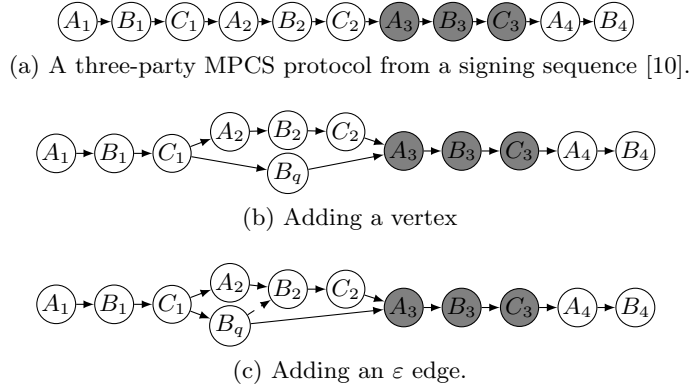
\begin{figure}[t]
\centering
\subfloat[\label{fig:classical3top}A three-party MPCS protocol from a signing sequence~\cite{KR12}.]{\qquad\tikzset{every picture/.style={scale=.5}}\begin{tikzpicture}[>=latex,line join=bevel]
\small
\begin{scope}
  \pgfsetstrokecolor{black}
  \definecolor{strokecol}{rgb}{1.0,1.0,1.0};
  \pgfsetstrokecolor{strokecol}
  \definecolor{fillcol}{rgb}{1.0,1.0,1.0};
  \pgfsetfillcolor{fillcol}
  \filldraw (0bp,-1bp) -- (0bp,22bp) -- (423bp,22bp) -- (423bp,-1bp) -- cycle;
\end{scope}
  \node (A_4) at (371bp,11bp) [draw,circle] {$A_4$};
  \node (A_3) at (251bp,11bp) [draw,fill=gray,circle] {$A_3$};
  \node (A_2) at (131bp,11bp) [draw,circle] {$A_2$};
  \node (A_1) at (11bp,11bp) [draw,circle] {$A_1$};
  \node (C_1) at (91bp,11bp) [draw,circle] {$C_1$};
  \node (C_3) at (331bp,11bp) [draw,fill=gray,circle] {$C_3$};
  \node (C_2) at (211bp,11bp) [draw,circle] {$C_2$};
  \node (B_1) at (51bp,11bp) [draw,circle] {$B_1$};
  \node (B_2) at (171bp,11bp) [draw,circle] {$B_2$};
  \node (B_3) at (291bp,11bp) [draw,fill=gray,circle] {$B_3$};
  \node (B_4) at (411bp,11bp) [draw,circle] {$B_4$};
  \draw [->] (B_1) ..controls (67.177bp,11bp) and (72.105bp,11bp)  .. (C_1);
  \draw [->] (C_3) ..controls (347.18bp,11bp) and (352.1bp,11bp)  .. (A_4);
  \draw [->] (C_1) ..controls (107.18bp,11bp) and (112.1bp,11bp)  .. (A_2);
  \draw [->] (B_3) ..controls (307.18bp,11bp) and (312.1bp,11bp)  .. (C_3);
  \draw [->] (C_2) ..controls (227.18bp,11bp) and (232.1bp,11bp)  .. (A_3);
  \draw [->] (A_3) ..controls (267.18bp,11bp) and (272.1bp,11bp)  .. (B_3);
  \draw [->] (A_1) ..controls (27.177bp,11bp) and (32.105bp,11bp)  .. (B_1);
  \draw [->] (A_4) ..controls (387.18bp,11bp) and (392.1bp,11bp)  .. (B_4);
  \draw [->] (A_2) ..controls (147.18bp,11bp) and (152.1bp,11bp)  .. (B_2);
  \draw [->] (B_2) ..controls (187.18bp,11bp) and (192.1bp,11bp)  .. (C_2);
\end{tikzpicture}

\qquad}\\
\subfloat[\label{fig:classical3withB}Adding a vertex]{\tikzset{every picture/.style={scale=.3}}\begin{tikzpicture}[>=latex,line join=bevel]
\small
\node (A_4) at (666bp,46bp) [draw,circle] {$A_4$};
  \node (A_3) at (450bp,46bp) [draw,fill=gray,circle] {$A_3$};
  \node (A_2) at (234bp,70bp) [draw,circle] {$A_2$};
  \node (A_1) at (18bp,43bp) [draw,circle] {$A_1$};
  \node (C_1) at (162bp,43bp) [draw,circle] {$C_1$};
  \node (C_3) at (594bp,46bp) [draw,fill=gray,circle] {$C_3$};
  \node (C_2) at (378bp,70bp) [draw,circle] {$C_2$};
  \node (B_q) at (306bp,18bp) [draw,circle] {$B_q$};
  \node (B_1) at (90bp,43bp) [draw,circle] {$B_1$};
  \node (B_2) at (306bp,72bp) [draw,circle] {$B_2$};
  \node (B_3) at (522bp,46bp) [draw,fill=gray,circle] {$B_3$};
  \node (B_4) at (738bp,46bp) [draw,circle] {$B_4$};
  \draw [->] (B_1) ..controls (118.31bp,43bp) and (130.93bp,43bp)  .. (C_1);
  \draw [->] (C_3) ..controls (622.31bp,46bp) and (634.93bp,46bp)  .. (A_4);
  \draw [->] (B_q) ..controls (350.32bp,26.542bp) and (401.53bp,36.64bp)  .. (A_3);
  \draw [->] (B_3) ..controls (550.31bp,46bp) and (562.93bp,46bp)  .. (C_3);
  \draw [->] (C_2) ..controls (405.75bp,60.829bp) and (419.47bp,56.125bp)  .. (A_3);
  \draw [->] (A_3) ..controls (478.31bp,46bp) and (490.93bp,46bp)  .. (B_3);
  \draw [->] (A_1) ..controls (46.306bp,43bp) and (58.926bp,43bp)  .. (B_1);
  \draw [->] (A_2) ..controls (262.31bp,70.78bp) and (274.93bp,71.141bp)  .. (B_2);
  \draw [->] (C_1) ..controls (189.75bp,53.317bp) and (203.47bp,58.61bp)  .. (A_2);
  \draw [->] (C_1) ..controls (206.32bp,35.373bp) and (257.53bp,26.357bp)  .. (B_q);
  \draw [->] (A_4) ..controls (694.31bp,46bp) and (706.93bp,46bp)  .. (B_4);
  \draw [->] (B_2) ..controls (334.31bp,71.22bp) and (346.93bp,70.859bp)  .. (C_2);
\end{tikzpicture}

}\\
\subfloat[\label{fig:classical3withBcausal}Adding an $\causal$ edge.]{\tikzset{every picture/.style={scale=.3}}\begin{tikzpicture}[>=latex,line join=bevel]
\small
\node (A_4) at (666bp,41bp) [draw,circle] {$A_4$};
  \node (A_3) at (450bp,41bp) [draw,fill=gray,circle] {$A_3$};
  \node (A_2) at (234bp,72bp) [draw,circle] {$A_2$};
  \node (A_1) at (18bp,45bp) [draw,circle] {$A_1$};
  \node (C_1) at (162bp,45bp) [draw,circle] {$C_1$};
  \node (C_3) at (594bp,41bp) [draw,fill=gray,circle] {$C_3$};
  \node (C_2) at (378bp,64bp) [draw,circle] {$C_2$};
  \node (B_q) at (234bp,18bp) [draw,circle] {$B_q$};
  \node (B_1) at (90bp,45bp) [draw,circle] {$B_1$};
  \node (B_2) at (306bp,64bp) [draw,circle] {$B_2$};
  \node (B_3) at (522bp,41bp) [draw,fill=gray,circle] {$B_3$};
  \node (B_4) at (738bp,41bp) [draw,circle] {$B_4$};
  \draw [->] (B_1) ..controls (118.31bp,45bp) and (130.93bp,45bp)  .. (C_1);
  \draw [->] (C_3) ..controls (622.31bp,41bp) and (634.93bp,41bp)  .. (A_4);
  \draw [->] (B_q) ..controls (291.38bp,24.059bp) and (387.46bp,34.386bp)  .. (A_3);
  \draw [->] (B_3) ..controls (550.31bp,41bp) and (562.93bp,41bp)  .. (C_3);
  \draw [->] (C_2) ..controls (405.6bp,55.26bp) and (419.08bp,50.831bp)  .. (A_3);
  \draw [->] (A_3) ..controls (478.31bp,41bp) and (490.93bp,41bp)  .. (B_3);
  \draw [->,dashed] (B_q) ..controls (260.77bp,34.933bp) and (276.4bp,45.207bp)  .. (B_2);
  \draw [->] (A_1) ..controls (46.306bp,45bp) and (58.926bp,45bp)  .. (B_1);
  \draw [->] (A_2) ..controls (261.94bp,68.921bp) and (274.64bp,67.469bp)  .. (B_2);
  \draw [->] (C_1) ..controls (189.75bp,55.317bp) and (203.47bp,60.61bp)  .. (A_2);
  \draw [->] (C_1) ..controls (189.75bp,34.683bp) and (203.47bp,29.39bp)  .. (B_q);
  \draw [->] (A_4) ..controls (694.31bp,41bp) and (706.93bp,41bp)  .. (B_4);
  \draw [->] (B_2) ..controls (334.31bp,64bp) and (346.93bp,64bp)  .. (C_2);
\end{tikzpicture}

}
\caption{\label{fig:classical3}Skeletal graphs of fair protocols
 (a, c) and an unfair protocol (b).}
\end{figure}

\begin{example}\label{ex:insufficient}
 The protocol in Figure~\ref{fig:classical3top} is fair by the results
 of~\cite{KR12}. By Theorem~\ref{thm:unfairness}, for every vertex $v\in\Sigset(\protspec)$ 
every permutation of 
signers in $\signers\setminus\set{\signer}$ occurs as a subsequence of
a path in $\seq(I,v)$. 
The protocol in Figure~\ref{fig:classical3withB} is obtained by
adding the vertex $B_q$ as a parallel thread of signer $B$. 
Thus the permutation property on the set of paths is preserved, yet
the protocol is not fair: An AC sequence is $(B_q, C_3, A_4|A_3)$. The
vertex $B_q$ is in $\Initset(\protspec)$, the evidence presented to the TTP
at $C_3$ includes the vertices causally preceding $C_2$, thus $B$ is
considered to be honest. The evidence presented by signer $A$ at $A_4$
are the ver\-ti\-ces causally preceding $A_3$ proving that $B$ is dishonest, but
$C$ is honest. Thus $A$ has sent a signature at $A_3$ but will
not receive signatures from $B$ and $C$.
\end{example}

If a protocol has no in-role parallelism, then the converse of
Theorem~\ref{thm:unfairness} is true. Thus we have a simple
criterion for the fairness of such protocols.
\begin{theorem}\label{thm:fairness-causal}
Let $\protspec$ be an optimistic DAG MPCS protocol without in-role parallelism.
Let $$I=\set{v\in\Initset(\protspec)\mid (v\prec w\wedge
 \nlabel(v)=\nlabel(w))\Rightarrow w\not\in\Initset(\protspec)}.$$ If all
paths from $I$ to $v\in\Sigset(\protspec)$ contain all permutations of
$\signers\setminus\set{\nlabel(v)}$ then $\protspec$ is fair for $\nlabel(v)$.
\end{theorem}

\begin{example}
By adding a causal edge between vertex $B_q$ and vertex $B_2$ of the
protocol in Figure~\ref{fig:classical3withB}, as shown in
Figure~\ref{fig:classical3withBcausal},
we obtain again a fair protocol.
\end{example}

\section{Protocols}\label{sec:protocols}

In this section we illustrate the theory and results obtained in the
preceding sections by proving optimality results and constructing a variety of protocols. 

\subsection{Minimal complexity}\label{sec:minimal_complexity}

We prove lower bounds for the two complexity measures defined in our
model, viz.~parallel and message complexity.

\begin{figure}[t]
 \centering
 \includegraphics[width=.4\textwidth]{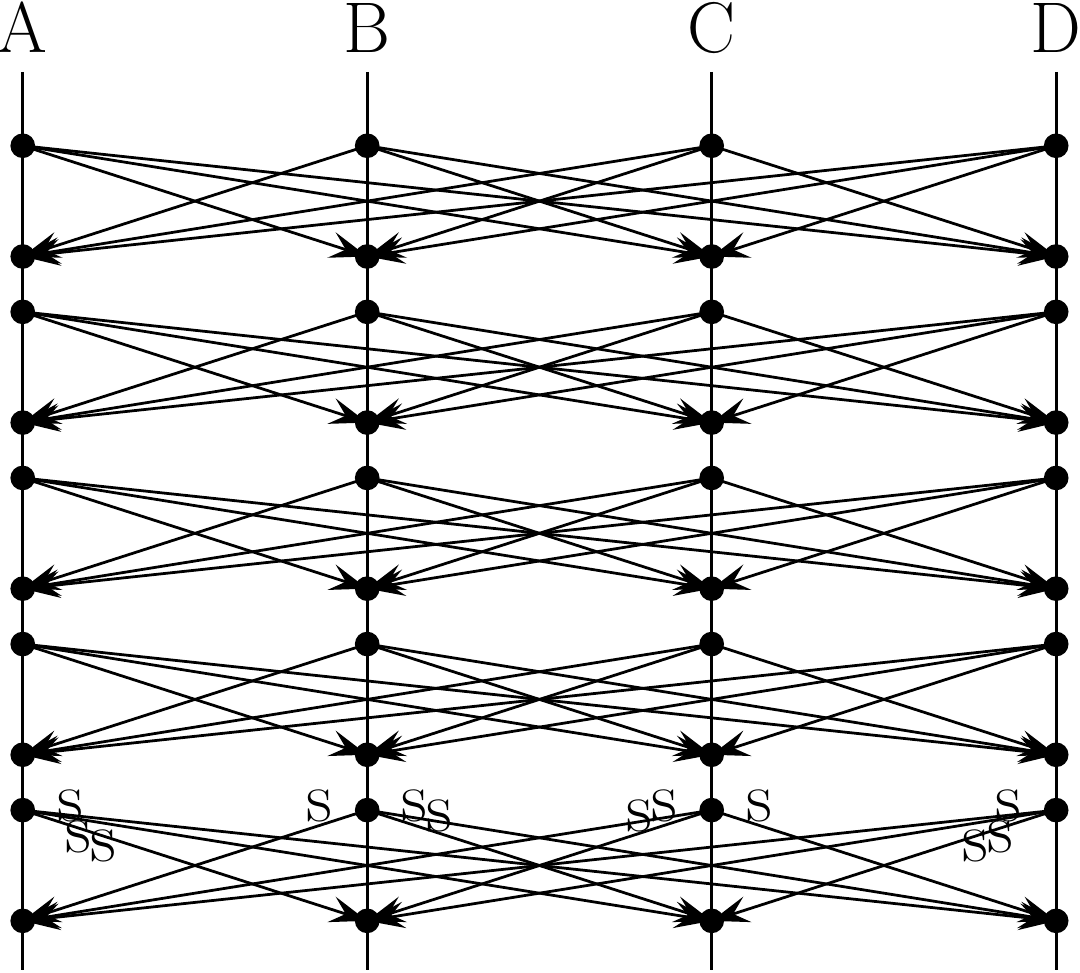}
 \caption{\label{fig:butterfly} A minimal
 4-party fair 
broadcasting protocol.}
\end{figure}

\begin{theorem}
The minimal parallel complexity for an optimistic fair DAG MPCS
protocol is $n+1$, where $n$ is the number of signers in the protocol.
\end{theorem}
\begin{proof}
 By Theorem~\ref{thm:unfairness}, every permutation of signers in the
 protocol must occur as a subsequence in the set of paths from a
 causally last vertex
 in the initial set to a vertex in the signing set. Since a last
 vertex $v$ in the initial set must have a non-empty knowledge
 $\knowledge(v)$, there must be a message edge causally preceding
 $v$. There are at
 least $n-1$ edges in the path between the vertices associated with the $n$
 signers in a permutation and there is at least one message edge
 outgoing from a vertex in the signing set. 
Thus a minimal length path for such a protocol must contain $n+1$ edges. 
\end{proof}

The minimal parallel complexity is attained by the broadcast protocols
of Baum-Waidner and Waidner~\cite{BW98}. An example 
is shown in Figure~\ref{fig:butterfly}.

\begin{theorem}\label{thm:minimal_message_complexity}
The minimal message complexity for an optimistic fair DAG MPCS
protocol is $\lambda(n)+2n-3$, where $n$ is the number of signers in
the protocol and $\lambda(n)$ is the
length of the shortest sequence which contains all permutations of
elements of an $n$-element set as subsequences.

The minimal message complexities for $2<n<8$ are $n^2+1$. 
The minimal message complexities for $n\geq 10$ are smaller or equal
to $n^2$.
\end{theorem}

\setlength{\intextsep}{2ex}

Note that while broadcasting protocols have a linear parallel complexity,
they have a cubic message complexity, since in each of the $n+1$
rounds each of the $n$ signers sends a message to every other signer.
Linear protocols on the other hand have quadratic minimal message
and parallel complexities. In the following we demonstrate that there
are DAG protocols which attain a linear parallel complexity while
maintaining a quadratic message complexity.

\subsection{Protocol constructions}\label{sec:constructions}

\paragraph{Single contractor, multiple subcontractors.}
A motivation for fair MPCS protocols given in~\cite{KR12} is a
scenario where a single entity, here referred to as a contractor,
would like to sign $k$ contracts with $k$ independent companies, in
the following referred to as subcontractors. The contractor has an
interest in either having all contracts signed or to not be bound by
any of the contracts. The subcontractors have no contractual
obligations towards each other. It would therefore be advantageous if
there is no need for the subcontractors to directly communicate with
each other.

The solutions proposed in~\cite{KR12} are
linear protocols. Their message and parallel complexities are thus
quadratic. Linear protocols can satisfy the requirement that
subcontractors do not directly communicate with each other only 
by greatly increasing the
message and parallel complexities. 

The protocol we propose here is a DAG, its
message complexity is $2(n+1)(n-1)$ 
and its parallel complexity is $2n+2$ for $n$ signers.
It therefore combines the low parallel complexity typically attained
by broadcasting protocols with the low message complexity of
linear protocols. 
Additionally, the protocol proposed does not require any direct
communication between subcontractors. 

\begin{figure*}[t]
\centering
\subfloat[
 \label{fig:contractor} A single contractor and three subcontractors.
]{
 \includegraphics[scale=0.39]{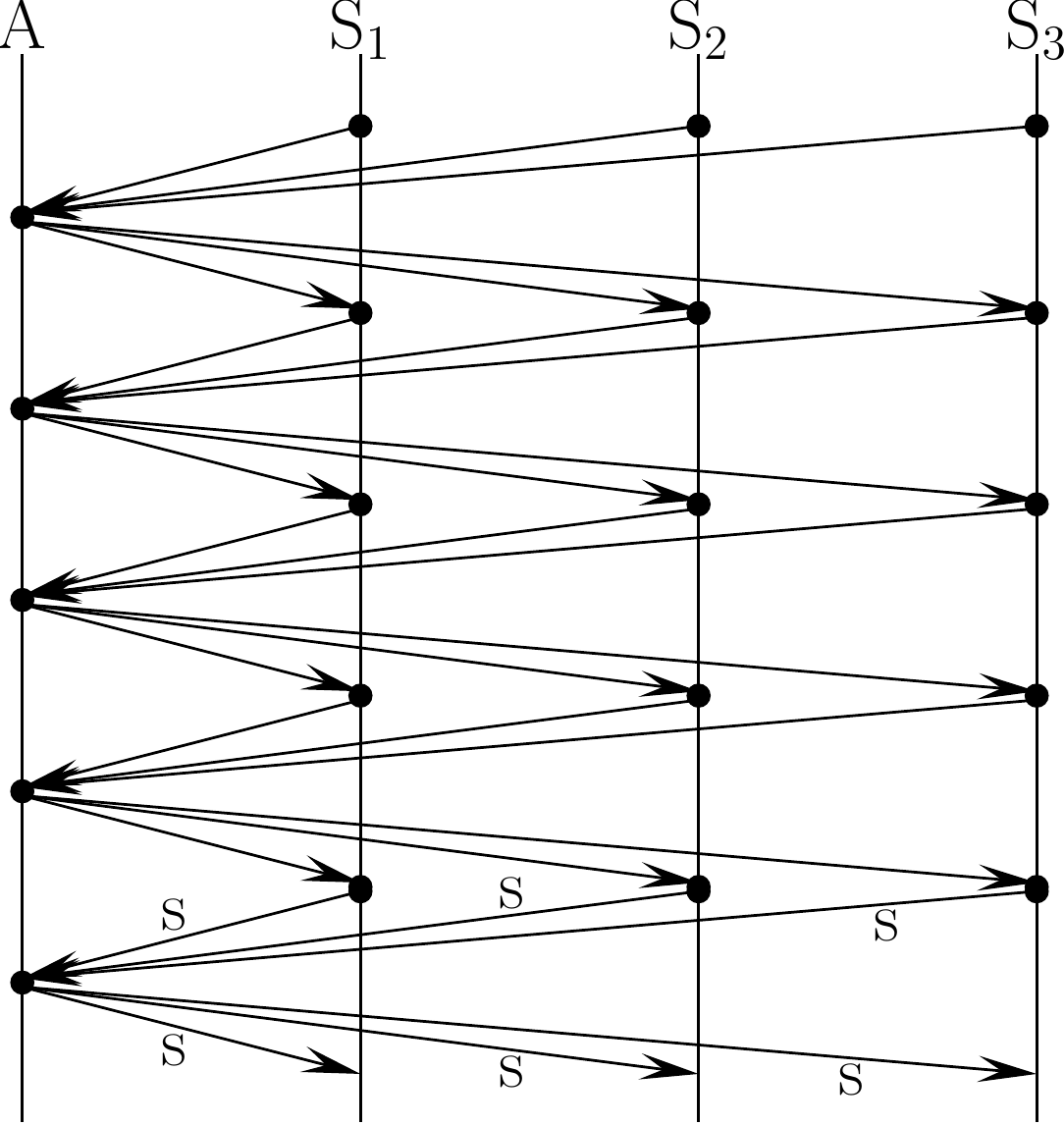}
}
\quad\quad\quad
\subfloat[
 \label{fig:two-twojoint}Two joint subcontractors.
]{
\includegraphics[scale=0.39]{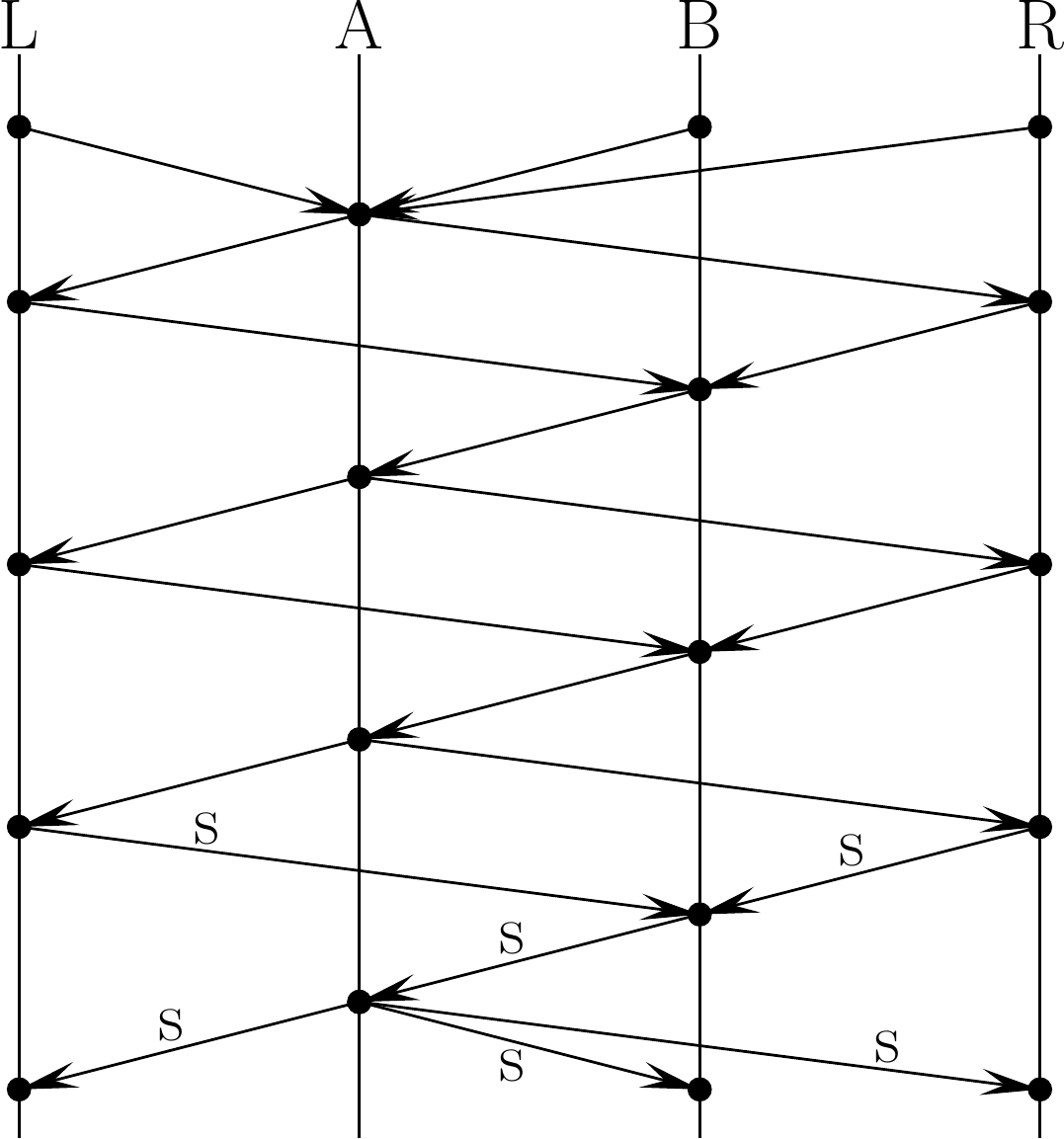}
}
\caption{Two examples of novel, fair DAG MPCS protocols.}
\end{figure*}

Figure~\ref{fig:contractor} shows a single contractor with three
subcontractors. The protocol can be subdivided into five rounds, one
round consisting of the subcontractors sending a message to the
contractor followed by the contractor sending a message to the subcontractors.
In the first four rounds promises are sent, in the final round
signatures are sent.
The protocol can be easily generalized to more than three
subcontractors. For every subcontractor added, one extra round of
promises needs to be included in the protocol specification.

The protocol is fair by
Theorem~\ref{thm:fairness-causal}. 
The MSC shown in Figure~\ref{fig:contractor}
resembles the skeletal graph from which it was built. The message
contents can be derived by 
computing the full graph according to
Condition~\ref{cond:transitivity} of
Definition~\ref{def:parallelMPCSspec}. The result is as follows. 
In each round of the protocol, each of the subcontractors sends to
the contractor a promise for the contractor and for each of the other
subcontractors. The contractor then sends to each of the
subcontractors all of the promises 
received and his own promise.
The final round is performed in the same manner, except that promises
are replaced by signatures.
 
\paragraph{Two contractors with joint subcontractors.}
Figure~\ref{fig:two-twojoint} shows a protocol where two contractors
want to sign a contract involving two subcontractors. The
subcontractors are independent of each other. 

After the initial step, where all signers send a promise to the first
contractor $A$, there are three protocol rounds, one
round consisting of the contractor $A$ sending promises to the two
subcontractors $L$ and $R$ in parallel which in turn send promises to
the second contractor $B$. 
A new round is started with the second contractor sending
the promises received with his own promise to contractor $A$.

This protocol, too, can be generalized to several independent
subcontractors. For every subcontractor added, one extra protocol round 
needs to be included in the protocol specification and each protocol
step of the subcontractors executed analogously.

\begin{figure}[!ht]
 \centering
 \includegraphics[scale=.4]{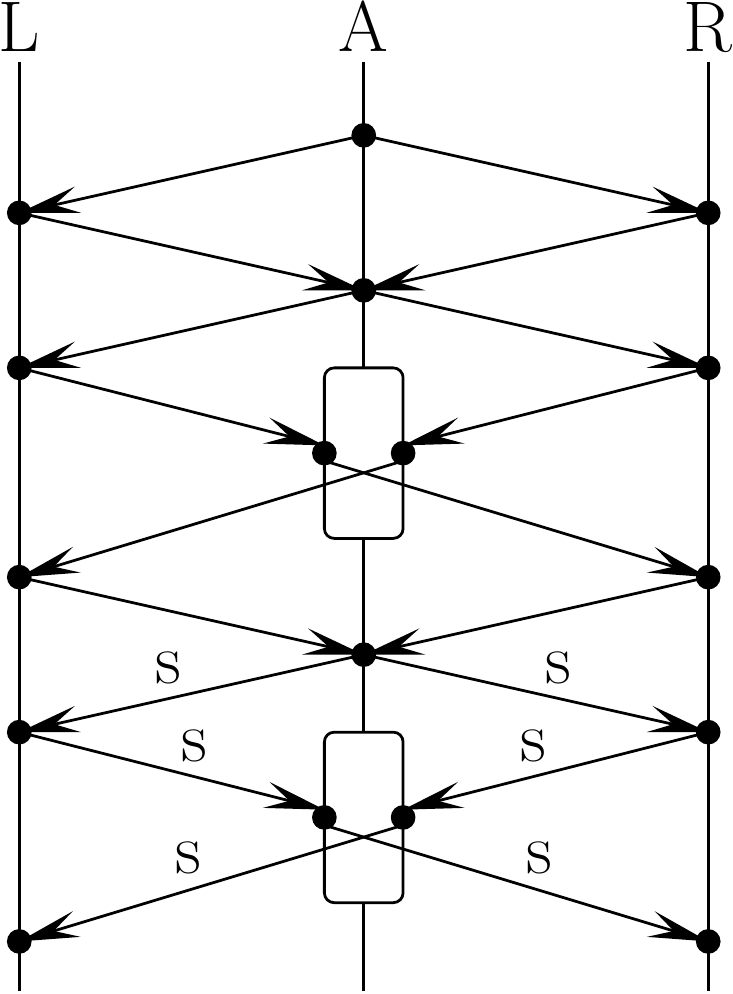}
 \caption{\label{fig:subthree-crossed} In-role parallelism.}
\end{figure}

\paragraph{Parallelism within a role.}
Figure~\ref{fig:subthree-crossed} shows an example of a subcontracting
protocol with in-role parallelism for the contractor role. The
contractor initiates the protocol. In the indicated parallel phase,
the contractor may immediately forward a promise by one of the
subcontractors along with his own promise to the other subcontractor
without waiting for the latter subcontractor's promise. The same is
true in the signing phase. The fairness property for this protocol has
been verified with a tool (described in Appendix~\ref{s:tool}) 
which tested fairness for each signer in all
possible executions. 

\section{Conclusion}

We have identified fair subcontracting as a challenging new problem
in the area of multi-party contract signing. We have made first steps
towards solving this problem by introducing DAG
MPCS protocols and extending existing fairness results from linear
protocols to DAG protocols.
For three typical subcontracting configurations 
we propose novel DAG MPCS protocols that perform well in terms
of message complexity and parallel complexity. 
Fairness of our
protocol schemes follows directly from our theoretical results
and we have verified it for concrete protocols with our
automatic tool.

There are a number of open research questions related to fair
subcontracting that we haven't addressed. We
mention two. 
The first concerns the
implementation of multi-contracts. In our current approach we consider
a single abstract contract shared by all parties. However, in practice
such a contract may consist of a number of subcontracts that are
accessible to the relevant signers only. How to cryptographically
construct such contracts and what information these
contracts should 
share is not evident.
Second, a step needs to be made towards putting our results into
practice. Given the application domains identified in this paper, we
must identify the relevant signing scenarios and topical boundary
conditions in order to develop dedicated protocols for each
application area.

\section*{Acknowledgement}
We thank Barbara Kordy for her many helpful comments on
this paper.

\appendix

\section{DAG MPCS Verification Tool}\label{s:tool}

We have developed a prototype tool in Python 2 that model checks 
a skeletal protocol graph for the fairness property (Definition~\ref{def:fair-execution})
in the execution model defined in Section~\ref{sec:spec_exec_model}.
The tool, along with specifications for the protocols presented
in this paper, is available at \url{http://people.inf.ethz.ch/rsasa/mpcs}.

The tool's verification procedure works directly on the
execution model and the TTP decision procedure (Algorithm~\ref{algo:resolve}). 
It therefore provides evidence for the correctness of the
protocols shown in Section~\ref{sec:protocols}, independent of 
the fairness proofs given in this paper.

The verification is performed as follows. For each specified signer,
the tool analyzes a set of executions in which the signer is honest
and all other signers dishonest. 
The tool does not analyze all possible executions. It starts the
analysis from the state where all
promises of dishonest signers have been sent, but no protocol step has
been performed by the honest signer. 
By analyzing this type of executions only, we do not miss any
attacks, because the honest signers' fairness is not invalidated
until he has sent a signature.
In this reduced set of executions, the dishonest signers retain the
possibility to contact the TTP from any of their vertices and all
these possibilities are explored by the tool.

We note that the same type of verification could be achieved with an
off-the-shelf model checker and we would expect better performance in
such a case. However, the code complexity and room for
error when encoding a given protocol and TTP decision procedure in a
model checker's input language is comparable to the code complexity of
this self-contained tool.

\section{Technical Details and Proofs}

\subsection{Technical Details}

\subsubsection{Parallelism within a role}

The MPCS protocols designed in this work allow for parallelism during
the execution of the protocol.
The specification language allows even for parallel threads to occur
within a signer role. 
This allows us to model the case where a signer role represents 
multiple branches of the same entity. A
signature issued by any branch represents the signature of the entire entity.
We expect that the signing processes across branches are not easily synchronizeable with each other. 
Such parallelism can be implemented in multiple ways. 
We discuss the various options and explain the choices made for this paper.

The first decision to be made is whether parallel threads of a
signer role should be assumed to have shared knowledge. 
In this paper, we choose the weaker assumption: memory for a signer's
parallel threads is local to the threads. 
This is in accordance with our expectation that parallel-threads are not easily synchronizeable
and allows us, for instance, to specify and analyze protocols in which 
representatives of a signing entity can
independently carry out parallel protocol steps without the need to
communicate and synchronize their combined knowledge. 
Causal dependence
between two actions of a signer is explicitly indicated in the
protocol specification.

This design decision leads to three options for handling protocol
failures.
\begin{enumerate}
\item All threads of a signer immediately synchronize and stop
executing whenever any of the threads intends to issue a resolve
request to the TTP. A designated resolution process per signer will be
required to continuously schedule all threads and take care of the
interaction with the TTP.
\item Threads of a signer contact the signer's designated resolution
process only when they intend to issue a resolve request. The
resolution process will take of contacting the TTP (only once per
signer) and distributing the TTP's reply upon request of the threads.
\item Threads of a signer are considered fully independent. A signer's
threads are not orchestrated. The TTP may take into account that two
requests can originate from the same signer, but from different
(causally not related) threads.
\end{enumerate}

In this paper we adopt the second option, which keeps the middle
between the fully synchronized and fully desynchronized model. This
will on the one hand allow for independent parallel execution of the
threads and on the other hand minimize the impact of the signer's
threading on the TTP's logic. From an abstract point of view, one
could even argue that the second and third option are equivalent if we
consider the signer's designated resolution processes just as part of
a distributed TTP.
We assume that the communication between a thread and the designated
resolution process is resilient.

\subsubsection{The class of DAG MPCS protocols}

The class of DAG MPCS protocols defined in Section~\ref{sec:MPCS} is 
restricted by condition~\ref{cond:transitivity} of
Definition~\ref{def:parallelMPCSspec}. It
requires that every signer $\signer$ sends a message to all
subsequent, causally following signers occurring before signer $\signer$'s
next step. 
While there are fair DAG MPCS protocols which do not belong to this
restricted class, such protocols are not going to have a lower
communication complexity. The reason for this is that each message
received by a signer serves as evidence for the TTP that the sender has executed
the protocol up to a certain step. Skipping such a message thus
lengthens the protocol, because the evidence is available only at a
later vertex. 

Furthermore, the restriction simplifies the reasoning about
fairness in that causal precedence $v\prec w$ between vertices $v, w$ is enough to
guarantee that there is a message sent from signer $\nlabel(v)$ to
signer $\nlabel(w)$ at some point between the execution of $v$ and the
execution of $w$. 
Finally, it also permits one to design, characterize, and 
represent protocols using skeletal graphs rather than
full graphs as displayed in Figure~\ref{fig:exampleDAG}.

\subsection{Proofs}

 The set
$\maxset(S) = \{v\in S\mid\forall w\in\vertices:
v\prec w\Rightarrow w\not\in S\}$ is the set of vertices in $S$ which
do not have any causally following vertices in $S$ and we will refer
to it as the set of \emph{maximal vertices} of $S$. 
Similarly, 
$\minset(S) = \{v\in S\mid\forall w\in\vertices:
w\prec v\Rightarrow w\not\in S\}$ 
is the set of vertices in $S$ which do not have
 any causally preceding vertices in $S$ and will be referred to as the 
set of \emph{minimal vertices} of $S$.

\noindent\textbf{Theorem.}
\emph{If there exists an AC sequence for a DAG MPCS protocol, then the
protocol is not fair.}

\begin{proof}
Conditions~\ref{PACcond:forcedAbort}
through~\ref{PACcondition4} imply that the TTP decision procedure
leads to an \abort{} token for the last signer to contact the TTP. The remaining two conditions imply that 
the last signer has sent a signature, but not received a
signature. 

To complete the proof, we need to construct an execution in which the
exit transitions occur in the order indicated by the AC sequence and 
signer $\nlabel(v_l)$ is honest.
Let $(v_1,\ldots,v_l\mid v)$ be an AC sequence.
For each vertex $v_i$ in the AC sequence, let $\overline{V_i}$ be the
causal closure of $\set{v,v_i}$ in $\vertices\cup\edges$. 
Note that the union of causally closed sets is causally closed. 
Let $\rho_i$ be the sequence of transitions
$\transition{
\bigcup_{j<i}\overline{V_j}}{\alpha}\dots\transition{}{\alpha'}\bigcup_{j\leq
 i}\overline{V_j}$ without exit transitions and such
that all states are causally closed. 

For $1\leq i\leq l$ and $\rho_i=[\state_0,\alpha_1,\ldots,\state_k]$, 
let $\rho_i'=[\state_0\cup\{(v_{1},v_\ttp),\dots, (v_{i-1},v_\ttp)\},$
$\alpha_1, \ldots, \state_k\cup\{(v_{1},v_\ttp),\dots,(v_{i-1},v_\ttp)\},\exit]$.
That is, $\rho_i'$ is equal to $\rho_i$, except for an additional exit
transition 
$\transition{\state}{\exit}{\state\cup\set{(v_i,v_\ttp)}}$ and
additional exit edges in all states which stem from exit transitions added to 
$\rho_{1}', \dots,\rho_{i-1}'$. 
Finally, for $\rho_l'=[\state_0,\alpha_1,\dots,\state_k,\exit]$, 
let 
$\rho_l'' =
[\state_0\setminus\set{v_l},\alpha_1,\dots,\state_k\setminus\set{v_l},\exit,
s_k\cup\set{(v_l,v_\ttp)}\setminus\set{v_l}]$. 

Then $\rho=\rho_1'\cdots\rho_{l-1}'\cdot\rho_l''$ is an execution in
which signer $\nlabel(v_l)$ is honest, since the restricted execution
is by construction causally closed in all states before the last state
and the single exit transition occurs in the last transition.

Unfairness for $\nlabel(v_l)$ follows since $\nlabel(v_l)$ has sent a
signature at $v$, not received all signatures from the other signers and received an $\abort$ from the TTP. 
\end{proof}

\begin{proof}[of Lemma \ref{lem:snd-count}]
Let $\rho=[\state_0,\alpha_1,\state_1,\dots,\alpha_l,\state_l]$ be an execution 
of $\protspec$. It is sufficient to show that if $\rho$
is closed, it contains all send events exactly once.
According to Definition~\ref{def:transition}, we know that for every 
$i\in\{0,\dots,l-1 \}$ we have 
$\transition{\state_i}{\alpha_{i+1}}{\state_{i+1}}\implies \state_i\not=\state_{i+1}$.
This implies that, in any execution, each step of the protocol 
(in particular every send event) can be executed at most once. 
Furthermore, if $\rho$ is closed, the third condition from 
Definition~\ref{def:exe-closed} implies that,
every send event has already occurred in $\rho$. Otherwise, there 
exists $e\in\edges$ such that $\elabel(e)=\snd$
and $\rho$ can be extended 
to $\rho\cdot[\snd,s_l\cup\{e\}]\in\exep$, 
which contradicts the closedness of $\rho$.
\end{proof}

\begin{lemma}\label{l:commonAncestorPerm}
 Let $\protspec$ be an optimistic fair DAG MPCS protocol. 
 Let $v, v', v''$ be pairwise distinct vertices assigned to the same signer 
 such that
\begin{enumerate}
 \setlength{\itemsep}{1pt}
 \setlength{\parskip}{0pt}
 \setlength{\parsep}{0pt}
\item $v\in\Sigset(\protspec)\setminus\Endset(\protspec)$, 
\item $v''$ is a maximal common ancestor of $v$ and $v'$, i.e., 
 $v''\prec w \prec v, v' \Rightarrow \nlabel(w)\neq\nlabel(v)$, and

\item for every signer $\signer\neq \nlabel(v)$ there exists a vertex
 $w\succ v''$ with $\nlabel(w)=\signer$.
\end{enumerate}
Then
for every permutation $(\signer_1,\ldots,\signer_{k-1})$ of
signers in $\signers\setminus\set{\nlabel(v)}$, there exists a sequence in 
$\seq(I,v'')$ which contains $(\signer_1,\ldots,\signer_{k-1})$ as a (not
necessarily consecutive) subsequence.
\end{lemma}

\begin{proof} 
Suppose there exists a permutation 
$(\signer_1,\ldots,\signer_{k-1})$ of signers in
$\signers\setminus\set{\nlabel(v)}$ which is not a subsequence of any
sequence in $\seq(I,v'')$. 
We construct an AC sequence as follows.
Let $V_1$ be the set of all vertices of $\signer_1$ in $I$.
For $i>1$, let $V_i$ be the minset of all vertices of $\signer_i$
which causally follow a vertex of $V_{i-1}$, i.e.
$V_i=\minset(\set{w\in\vertices\mid \nlabel(w)=\signer_i\wedge
\exists w'\in V_{i-1}: w'\prec w})$. 
Since for every signer there exists a vertex which causally follows
$v''$, it follows that for some $j$, there exists a vertex $v_j\in V_j$
with $v''\prec v_j$. (Else we have contradiction to the assumption that 
$(\signer_1,\ldots,\signer_{k-1})$ is a missing permutation in $\seq(I,v'')$.)
Thus we obtain a sequence $(v_1,\ldots,v_j,w|v)$, where $v''\prec
w\preceq v'$, $\nlabel(w)=\nlabel(v)$, 
which is an AC sequence. 
\end{proof}

\begin{proof}[of Theorem~\ref{thm:unfairness}]
It suffices to verify the statement for a subset of all vertices in
$\Sigset(\protspec)$ by the following two facts:
\textbf{Fact 1:} Let $v\in\Sigset(\protspec)$ be a causally earliest vertex of a signer from which a signature is
sent, i.e. $\forall w\in\Sigset(\protspec): w\prec v\Rightarrow \nlabel(w)\neq\nlabel(v)$.
If $\seq(I,v)$ contains all permutations of signers in
$\signers\setminus\set{\nlabel(v)}$, then 
$\seq(I,w)$ contains all such permutations of signers for all $w\succ v$
with $\nlabel(w)=\nlabel(v)$.
\textbf{Fact 2:} If $v\in\Sigset(\protspec)$ such that for every signer
$\signer\in\signers\setminus\set{\nlabel(v)}$ there is a vertex
$w\prec v$ for which $\seq(I,w)$ contains all permutations of signers
in $\signers\setminus\set{\nlabel(w)}$, then $\seq(I,v)$ contains all
permutations of signers in $\signers\setminus\set{\nlabel(v)}$.

Thus, we may assume that $v\in\Sigset(\protspec)$ is a causally earliest vertex
of a signer from which a signature is sent (by Fact 1) and that
$v\not\in\Sigset(\protspec)\setminus\Endset(\protspec)$ (by Fact 2).

Since $v$ is a causally earliest vertex of a signer from which a
signature is sent, it follows by the fact that the protocol is
optimistic that for every signer other than $\nlabel(v)$
there exists a vertex which causally
follows $v$ or that there exists another vertex $v''$ of signer
$\nlabel(v)$ from which a signature is sent such that
$v''\npreceq v$ and $v\npreceq v''$. 
We consider these two cases separately.

\begin{enumerate}
\item For every signer other than $\nlabel(v)$, there exists a vertex
 which causally follows $v$.

We split this case into two separate subcases depending on whether
there exists a vertex $v'$ of signer $\nlabel(v)$ which causally follows
$v$.
\begin{enumerate}
\item $\exists v'\succ v: \nlabel(v')=\nlabel(v)$.
Let $(\signer_1,\ldots,\signer_{k-1})$ be a permutation of signers in
$\signers\setminus\set{\nlabel(v)}$ and 
suppose towards a contradiction that the permutation 
does not appear as a subsequence of any sequence in
$\seq(I,v)$. We construct an AC sequence as follows.
Let $V_1$ be the set of all vertices of $\signer_1$ in $I$.
For $i>1$, let $V_i$ be the minset of all vertices of $\signer_i$
which causally follow a vertex of $V_{i-1}$, i.e.
$V_i=\minset(\{w\in\vertices\mid \nlabel(w)=\signer_i\wedge
\exists w'\in V_{i-1}: w'\prec w\})$. 

Since for every signer there exists a vertex which causally follows
$v$, it follows that for some $j$ there exists a vertex $v_j\in V_j$
with $v\prec v_j$, else we have contradiction to the assumption that 
the permutation $(\signer_1,\ldots,\signer_{k-1})$ is not a
subsequence of any sequence in $\seq(I,v)$.

By construction, there exists a vertex in $V_{j-1}$ which causally
precedes $v_j$ and thus we obtain a sequence $(v_1,\ldots,v_j,v'|v)$
which is an AC sequence. 
\item $\neg\exists v'\succ v: \nlabel(v')=\nlabel(v)$.

Since the protocol is optimistic, there exists a vertex assigned to
signer $\nlabel(v)$ such that 
$v'\in\Endset(\protspec)$. Since $v\not\in\Endset(\protspec)$, it
follows that $v'$ is not causally related to $v$. By the remark
preceding Lemma~\ref{l:commonAncestorPerm}, 
there exists a common ancestor $v''$ or $v$
and $v'$ and $v, v', v''$ satisfy the hypothesis of the Lemma.
Thus
there exists a vertex $w$ causally
preceding $v$ such that 
$\seq(I,w)$ contains all permutations of signers in
$\signers\setminus\set{\nlabel(v)}$ and therefore $\seq(I,v)$ contains
all such permutations. 
\end{enumerate}
\item There are causally unrelated vertices of signer $\nlabel(v)$
 from which signatures are sent.

Let $v'\neq v$ be such a vertex.
By Equation~\eqref{condition1} in Section~\ref{sec:suffnec}, 
there is a vertex $w$ assigned to signer $\nlabel(v)$ 
which causally precedes all vertices of $\nlabel(v)$ which are in
$\Sigset(\protspec)$. 
Let $v''$ be a maximal such vertex, i.e. for any vertex $w'$ assigned to signer
$\nlabel(v)$, there exists a vertex in $\Sigset(\protspec)$ of signer
$\nlabel(v)$ which does not causally follow $v''$.

Since the protocol is optimistic, for every signer $\signer$ in the protocol,
there exists a vertex $w''$, $\nlabel(w'')=\signer$ which causally
follows $v''$. 

Then the vertices $v, v', v''$ satisfy the hypothesis of 
Lemma~\ref{l:commonAncestorPerm}, thus
there exists a vertex $w$ causally
preceding $v$ such that 
$\seq(I,w)$ contains all permutations of signers in
$\signers\setminus\set{\nlabel(v)}$ and therefore $\seq(I,v)$ contains
all such permutations. 
\end{enumerate}
\end{proof}

\begin{proof}[of Theorem~\ref{thm:fairness-causal}]
 Suppose that the protocol is not fair. Consider a shortest AC sequence,
 $(v_1,\ldots,v_l|\_)$, $\nlabel(v_l)=\nlabel(v)$. Since the sequence is
 a shortest sequence, we have that $v_2\not\in\Initset(\protspec)$, else
 $(v_2,\ldots,v_l|\_)$ would be a shorter AC sequence. Consider the
 permutation of signers $(\signer_1,\ldots,\signer_l)$ corresponding
 to the AC sequence, i.e. $\signer_i = \nlabel(v_i)$. 

Let $w_l$ be the unique vertex in 
$$\minset(\set{w\in\Sigset(\protspec)\mid \nlabel(w)=\nlabel(v_l)}).$$ Existence
of a vertex in the set 
follows from the fact that the protocol is optimistic, uniqueness
follows from the fact that there is no in-role parallelism,
i.e.\ the vertices assigned to a particular signer are totally ordered.
By hypothesis, the set of paths from $I$
to $w_l$ contains all permutations of signers
$\signers\setminus\set{\nlabel(v)}$. Let $u_1,\ldots,u_l$ be the
vertices associated with one such permutation. Note that either
$u_1\in I$ or we can find $u_1'\in I$, $u_1'\prec u_1$ and
$\nlabel(u_1')=u_1$. Thus we may assume $u_1\in I$. 
We have $u_l\preceq w_l\prec v_l$. 
We also have $w_l\prec v_{l-1}\prec v_l$, else
condition~\ref{PACcondition4} for $(v_1,\ldots,v_l|\_)$ being an AC sequence (Definition~\ref{def:PACseq}) would be violated.

Thus we have
$u_l\preceq w_l\prec v_{l-1}\prec v_l$. 
This forms the basis for
the inductively constructed sequence $w_1,\ldots, w_{l}$:
Given $w_{i+1},\ldots,w_l$, satisfying $u_{i+1}\preceq w_{i+1}\prec v_i\prec v_{i+1}$, 
let $w_i$ be the unique vertex in 
$\maxset(\set{w\prec w_{i+1}\mid
 \nlabel(w)=\nlabel(v_i)})$. Existence of a vertex in the set follows from
$u_i\prec u_{i+1}\preceq w_{i+1}$ and 
 uniqueness follows from the lack of 
in-role parallelism. By construction, $u_{i}\preceq w_{i}\prec
v_{i}$. If $i>1$, then we also have 
$u_{i}\preceq w_{i}\prec v_{i-1}\prec v_{i}$, else
condition~\ref{PACcondition4} for $(v_1,\ldots,v_l|\_)$ being an AC sequence (Definition~\ref{def:PACseq}) would be violated.

Thus, we have constructed a sequence $w_1,\ldots,w_l$ satisfying
$u_1\preceq w_1\prec w_2\prec v_1\prec v_2$. This is not possible, since 
$\nlabel(u_1)=\nlabel(v_1)$ and $u_1\prec v_1\in\Initset(\protspec)$,
contradicting $u_1\in I$.
\end{proof}

\begin{lemma}\label{l:singleComponent}
Let $G=(\vertices,\edges)$ be the DAG of a fair optimistic DAG
MPCS protocol for two or more signers. Let $G'=(\vertices',\edges')$, where
$\vertices'=\vertices\setminus\set{v_\ttp}$ and
$\edges'=\edges\setminus\set{(v,w)\in E\mid v=v_\ttp\vee w=v_\ttp}$,
be the DAG obtained by removing the TTP vertex and corresponding
edges. 
Then $G'$ is a single connected component.
\end{lemma}

\begin{proof}
Suppose there are more than one connected components in $G'$. Let
$v\in\Sigset(\protspec)$ be a causally earliest vertex from which a signature is
sent, i.e. $\forall w\in\vertices: w\prec v\Rightarrow w\not\in\Sigset(\protspec)$.

Let $w$ be a vertex in the $\Initset(\protspec)$ of a different connected component than $v$.
We have two cases:
\begin{itemize}
\item $\nlabel(w)=\nlabel(v)$. Then $(w|v)$ is an AC sequence. 
\item $\nlabel(w)\neq\nlabel(v)$. Let $v'$ be a vertex such that
 $\nlabel(v)=\nlabel(v')$ and $v'\npreceq v$. Such a vertex exists,
 because the protocol is optimistic, thus there must be a vertex of
 signer $\nlabel(v)$ receiving a signature. But such a vertex cannot
 precede $v$, because
 $v$ is a causally earliest vertex from which a signature is sent.

 Consider two cases:
 \begin{itemize}
 \item $w\nprec v'$: Then $(w,v'|v)$ is an AC sequence. 
 \item $w\prec v'$: Then $w$ and $v'$ are in the same connected
 component and $v$ is in another connected component. If
 $v'\not\in\Initset(\protspec)$, then let $v''\prec v'$ be
 such that $\nlabel(v')=\nlabel(v'')$ and
 $v''\in\Initset(\protspec)$. Else let $v''=v'$.

 Then $(v''|v)$ is an AC sequence. 
 \end{itemize}
\end{itemize}
\end{proof}

\begin{proof}[of Theorem \ref{thm:minimal_message_complexity}]
 The minimal message complexity has been derived for optimistic fair
 linear protocols
 in~\cite{MRT09,KR12}. 
 Since these protocols are a subset of DAG MPCS protocols
 we see that the same
 message complexity can be attained. 
We need to show that there are no optimistic DAG MPCS protocols
with lower message complexity. 
 By Theorem~\ref{thm:unfairness}, every permutation of signers in the
 protocol must occur as a subsequence in the set of paths from a
 maximal vertex of the set of vertices of a signer
 in the initial set to a vertex in the signing set. 

 Consider any fair optimistic DAG MPCS protocol $\protspec=(V,E,\nlabel,\elabel,\ttpfun)$.
 Construct a linear DAG $(V,E')$ by choosing any topologically sorted
 list $(v_1,\ldots,v_k)$ of the vertices in $(V,E)$ and 
 setting $E'=\set{(v_i,v_{i+1}) | 1\leq i\leq k}$.
 Since all permutations of signers occur
 along the paths in the DAG $(V,E)$ under the labelling
 $\nlabel:\vertices\to\signers$, they also occur in the topologically 
 sorted list $(v_1,\ldots,v_k)$ under the same labelling. Since the 
 DAG is a single connected component by Lemma~\ref{l:singleComponent},
 the number of edges in $E'$ is smaller or equal to the number of
 edges in $E$. Thus the message complexity of $\protspec$ is greater
 than or equal
 to the message complexity of a protocol based on the linear DAG $(V,E')$.

 The specific numbers for message complexity follow from~\cite{KR12,R12}. 
\end{proof}

\end{document}